\documentclass{llncs}
\usepackage{t1enc}
\usepackage{times}
\usepackage{amsfonts}

\newcommand{\set}[1]{\left\{ #1\right\}}
\newcommand{\gilt}{:}
\newcommand{\sodass}{\,:\,}
\newcommand{\setGilt}[2]{\left\{ #1\sodass #2\right\}}

\newcommand{\realrange}[2]{\left[#1, #2\right]}

\newcommand{\unitrange}[2]{\realrange{0}{1}}

\newcommand{\llabel}[1]{\label{\labelprefix:#1}}
\newcommand{\labelprefix}{} 

\newcommand{\discussionsize}{\small}

\newcommand{\frage}[1]{}

\newenvironment{code}{\noindent
\begin{tabbing}%
\hspace{2em}\=\hspace{2em}\=\hspace{2em}\=\hspace{2em}\=\hspace{2em}\=%
\hspace{2em}\=\hspace{2em}\=\hspace{2em}\=\hspace{2em}\=\hspace{2em}\=%
\kill}{\end{tabbing}}

\newcommand{\labelcommand}{}
\newcommand{\captiontext}{}
\newsavebox{\codeparam}
\newcounter{lineNumber}
\newenvironment{disscodepos}[3]{%
\renewcommand{\labelcommand}{#2}%
\renewcommand{\captiontext}{#3}%
\sbox{\codeparam}{\parbox{\textwidth}{#3}}%
\begin{figure}[#1]\begin{center}\begin{code}\setcounter{lineNumber}{1}}{%
\end{code}\end{center}\caption{\llabel{\labelcommand}\captiontext}\end{figure}}

{\end{disscodepos}}

\newcommand{\Is}       {:=}

\newdimen\endofsize\endofsize=0.5em
\def\endofbeweis{~\quad\hglue\hsize minus\hsize
                 \hbox{\vrule height \endofsize width
\endofsize}\par}

\usepackage{epsfig}
\usepackage{url}
\usepackage{amsmath}
\usepackage{amssymb}
\usepackage{pdflscape}
\usepackage{latexsym}
\usepackage{todonotes}
\usepackage{algorithmic}
\usepackage{multicol}
\usepackage{wrapfig}
\usepackage{numprint}
\usepackage{theorem}
\usepackage{makeidx}

\npdecimalsign{.} 

\def\MdR{\ensuremath{\mathbb{R}}}

\newcommand{\innerOuter}{\mathrm{innerOuter}}
\newcommand{\expansion}{\mathrm{expansion}}
\newcommand{\algname}{KaFFPa}

\newcommand{\Id}[1]{\ensuremath{\mathit{#1}}}
\usepackage{color}
\newcommand{\AuthorChris}[1]{#1}

\newcommand{\mytitle}{Engineering Multilevel Graph Partitioning Algorithms}
\begin{document}
\title{\mytitle}
\author{
Peter Sanders,  Christian Schulz 
}
\institute{Karlsruhe Institute of Technology (KIT), 76128 Karlsruhe, Germany\\ 
\email{\{sanders,christian.schulz\}@kit.edu}}
\date{today}

\maketitle
\begin{abstract}
We present a multi-level graph partitioning algorithm using novel local improvement algorithms and global search strategies transferred from multigrid linear solvers. 
Local improvement algorithms are based on max-flow min-cut computations and more localized FM searches. 
By combining these techniques, we obtain an algorithm that is fast on the one hand and on the other hand is able to improve the best known partitioning results for many inputs. 
For example, in Walshaw's well known benchmark tables we achieve 317 improvements for the tables 1\%, 3\% and 5\% imbalance. Moreover, in 118 out of the 295 remaining cases we have been able to reproduce the best cut in this benchmark.  
\end{abstract}

\section{Introduction}
\frage{todo: rewrite depending on where the main improvements are}
\textit{Graph partitioning} is a common technique in computer science, engineering, and related fields.
For example, good partitionings of unstructured graphs are very valuable in the area of \emph{high performance computing}.
In this area graph partitioning is mostly used to partition the underlying graph model of computation and communication.
Roughly speaking, vertices in this graph represent computation units and edges denote communication. 
Now this graph needs to be partitioned such there are few edges between the blocks (pieces). 
In particular, if we want to use $k$ PEs (processing elements) we want to partition the graph into $k$ blocks of about equal size. 
In this paper we focus on a version of the problem that constrains the
maximum block size to $(1+\epsilon)$ times the average block size and tries to
minimize the total cut size, i.e., the number of edges that run between blocks.

A successful heuristic for partitioning large graphs is the \emph{multilevel graph partitioning} (MGP) approach depicted in Figure~\ref{fig:mgp}
where the graph is recursively \emph{contracted} to achieve smaller graphs which should reflect the same basic structure as the initial graph. After applying an \emph{initial partitioning} algorithm to the smallest graph, the contraction is undone and, at each level, a
\emph{local refinement} method is used to improve the partitioning induced by the coarser level. 

Although several successful multilevel partitioners have been developed in the last 13 years, we had the impression that certain\frage{ps was: many. which sounds like we already did in the previous paper what we do here} aspects of the method are not well understood. 
We therefore have built our own graph partitioner KaPPa \cite{kappa} (Karlsruhe Parallel Partitioner) with focus on scalable parallelization. Somewhat astonishingly, we also obtained improved partitioning quality through rather simple methods. This motivated us to make a fresh start\frage{ps was: a step backwards. This sounds negative} putting all aspects of MGP on trial\frage{moved this phrase forward and dropped later. Things sounded repetitive.}. 
Our focus is on solution quality and sequential speed for large graphs. We defer the question of parallelization since it introduces complications that make it difficult to try out a large number of alternatives
for the remaining aspects of the method. This paper reports the first results 
we have obtained which relate to the local improvement methods and overall search strategies. We obtain a system that can be configured to either achieve the best known partitions for many standard benchmark instances or to be the fastest available system for large graphs while still improving partitioning quality compared to the previous fastest system.\frage{todo: match with overview, remove redundancies.}
\begin{figure}[t]
\begin{center}
\includegraphics[width=0.4\textwidth]{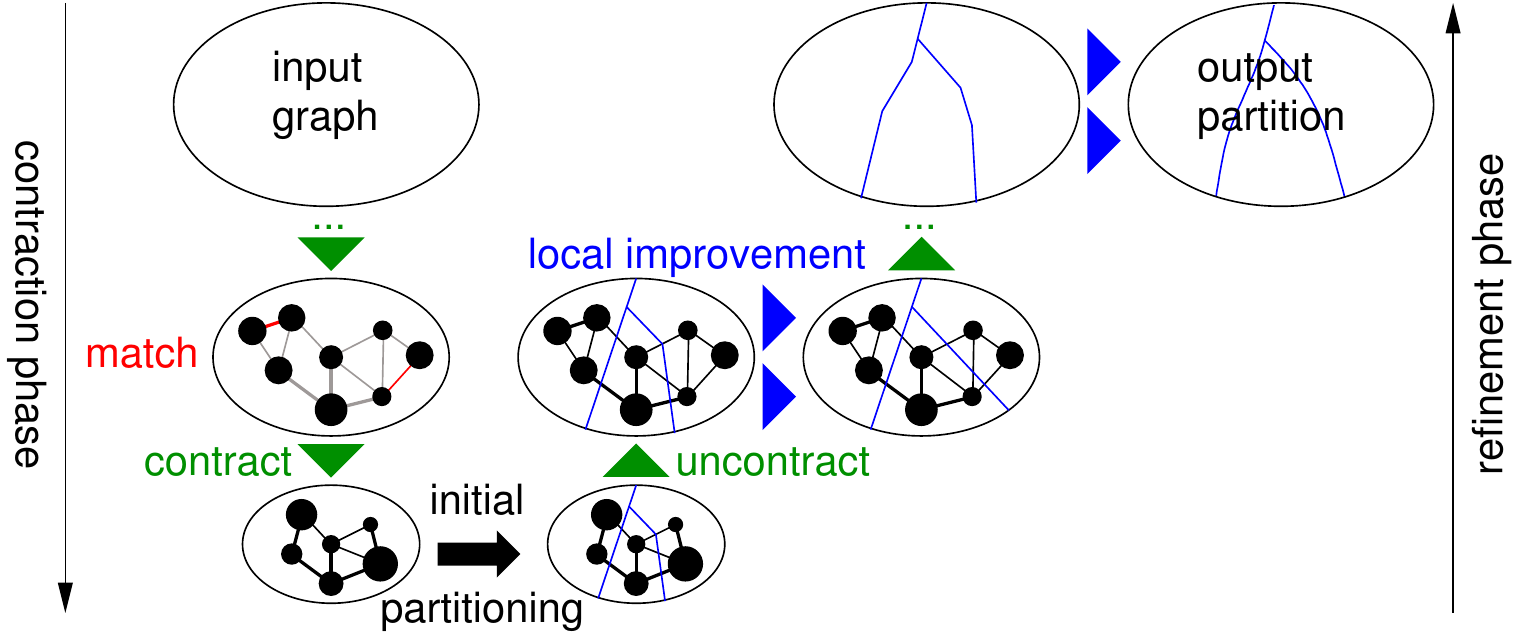}
\end{center}
\caption{Multilevel graph partitioning.}
\label{fig:mgp}
\end{figure}

We begin in Section~\ref{s:preliminaries} by introducing basic concepts. 
After shortly presenting Related Work in Section~\ref{s:related} we continue describing novel local improvement methods in Section~\ref{s:refinement}. 
This is followed by Section~\ref{s:globalsearch} where we present new global search methods. 
Section~\ref{s:experiments} is a summary of extensive experiments done to tune the algorithm and evaluate
its performance.  
We have implemented these techniques in the graph partitioner KaFFPa (Karlsruhe Fast Flow Partitioner) which is written in C++. Experiments reported in Section~\ref{s:experiments} indicate that KaFFPa scales well
        to large networks and is able to compute partitions of very high quality.

\section{Preliminaries}\label{s:preliminaries}
\subsection{Basic concepts}
Consider an undirected graph $G=(V,E,c,\omega)$ 
with edge weights $\omega: E \to \MdR_{>0}$, node weights
$c: V \to \MdR_{\geq 0}$, $n = |V|$, and $m = |E|$.
We extend $c$ and $\omega$ to sets, i.e.,
$c(V')\Is \sum_{v\in V'}c(v)$ and $\omega(E')\Is \sum_{e\in E'}\omega(e)$.
$\Gamma(v)\Is \setGilt{u}{\set{v,u}\in E}$ denotes the neighbors of $v$.

We are looking for \emph{blocks} of nodes $V_1$,\ldots,$V_k$ 
that partition $V$, i.e., $V_1\cup\cdots\cup V_k=V$ and $V_i\cap V_j=\emptyset$
for $i\neq j$. The \emph{balancing constraint} demands that 
$\forall i\in 1..k\gilt c(V_i)\leq L_{\max}\Is (1+\epsilon)c(V)/k+\max_{v\in V} c(v)$ for
some parameter $\epsilon$. 
\AuthorChris{The last term in this equation arises because each node is atomic and
therefore a deviation of the heaviest node has to be allowed.}
The objective is to minimize the total \emph{cut} $\sum_{i<j}w(E_{ij})$ where 
$E_{ij}\Is\setGilt{\set{u,v}\in E}{u\in V_i,v\in V_j}$. 
An abstract view of the partitioned graph is the so called \emph{quotient graph}, where vertices represent blocks and edges are induced by connectivity between blocks. An example can be found in Figure~\ref{fig:pictureoverview}.
By default, our initial inputs will have unit edge and node weights. 
However, even those will be translated into weighted problems in the course of the algorithm.

A matching $M\subseteq E$ is a set of edges that do not share any common nodes,
i.e., the graph $(V,M)$ has maximum degree one.  \emph{Contracting} an edge $\set{u,v}$ means to replace the nodes $u$ and $v$ by a
 new node $x$ connected
to the former neighbors of $u$ and $v$. We set $c(x)=c(u)+c(v)$ 
\AuthorChris{so the weight of a node at each level is the number of nodes it is representing in the original graph}. If replacing
edges of the form $\set{u,w},\set{v,w}$ would generate two parallel edges
$\set{x,w}$, we insert a single edge with
$\omega(\set{x,w})=\omega(\set{u,w})+\omega(\set{v,w})$.

\emph{Uncontracting} an edge $e$ undos its contraction. 
In order to avoid tedious notation, $G$ will denote the current state of the graph
before and after a (un)contraction unless we explicitly want to refer to 
different states of the graph.

The multilevel approach to graph partitioning consists of three main phases.
In the \emph{contraction} (coarsening) phase, 
we iteratively identify matchings $M\subseteq E$ 
and contract the edges in $M$. This is repeated until $|V|$ falls below some threshold.
Contraction should quickly reduce the size of the input and each computed level
should reflect the global structure of the input network. In particular,
nodes should represent densely connected subgraphs.
\frage{hier bereits edge ratings und matching algs sowie time/quality tradeoffs diskutieren?}

Contraction is stopped when the graph is small enough to be directly
partitioned in the \emph{initial partitioning phase} using some other algorithm. 
We could use a trivial initial
partitioning algorithm if we contract until exactly $k$ nodes are left. However,
if $|V|\gg k$ we can afford to run some expensive algorithm for initial
partitioning.

In the \emph{refinement} (or uncoarsening) phase, the matchings are iteratively uncontracted.  After uncontracting a matching, the refinement algorithm moves nodes between blocks in order to improve the cut size or balance.  The nodes to move are often found using some kind of local search. The intuition behind this approach is that a good partition at one level of the hierarchy will also be a good partition on the next finer level so that refinement will quickly find a good solution.
\subsection{More advanced concepts}
\frage{ps ggf ganz weglassen und verbleibende Teile nach oben schieben
oder nur kurz unter implementation erwähnen?}
 This section gives a brief overview over the algorithms KaFFPa uses during contraction and initial partitioning. 
KaFFPa makes use of techniques proposed in \cite{kappa} namely the application of edge ratings, the GPA algorithm to compute high quality matchings, pairwise refinements between blocks and it also uses Scotch \cite{Scotch} as an initial partitioner \cite{kappa}.

\paragraph*{Contraction}\label{p:morecontraction} The contraction starts by rating the edges using a \emph{rating function}. The rating function indicates how much sense it makes to contract an edge
based on \emph{local} information.  Afterwards a \emph{matching} algorithm tries to
maximize the sum of the ratings of the contracted edges looking at the \emph{global} structure of the graph. While
the rating functions allows us a flexible characterization of what a ``good''
contracted graph is, the simple, standard definition of the matching problem
allows us to reuse previously developed algorithms for weighted matching. Matchings are contracted until
the graph is ``small enough''.
We employed the ratings  
\newcommand{\Outer}{\mathrm{Out}}
$\expansion^{*2}(\set{u,v})\Is {\omega(\set{u,v})^2}/{c(u)c(v)}$ and
$\innerOuter(\set{u,v})\Is {\omega(\set{u,v})}/{(\Outer(v)+\Outer(u)-2\omega({u,v}))}$
where $\Outer(v)\Is \sum_{x\in\Gamma(v)}\omega(\set{v,x})$, since they yielded the best results in \cite{kappa}. 
As a further measure to avoid unbalanced inputs to the initial partitioner, KaFFPa never allows a node $v$ to participate in a contraction if the weight of $v$ exceeds $1.5n/20k$

We used the \textit{Global Path Algorithm (GPA)} which runs in near linear time to compute matchings. 
The Global Path Algorithm was proposed in \cite{MauSan07} as a synthesis of the Greedy algorithm and the Path Growing Algorithm \cite{DH03a}. 
It grows heavy weight paths and even length cycles to solve the matching problem on those optimally using dynamic programming. 
We choose this algorithm since in \cite{kappa} it gives empirically considerably better results than Sorted Heavy Edge Matching, Heavy Edge Matching or Random Matching \cite{SchKarKum00}.

Similar to the Greedy approach, GPA scans the edges in order of decreasing weight
but rather than immediately building a matching, it first constructs a collection
of paths and even length cycles. Afterwards, optimal solutions are computed for each
of these paths and cycles using dynamic programming. 
\paragraph*{Initial Partitioning} The contraction is \textit{stopped} when the number of remaining nodes is
below 
$\max{(60k,n/(60k))}$. The graph is then small enough
to be initially partitioned by some other partitioner.
Our framework allows using kMetis or Scotch for initial
partitioning. As observed in \cite{kappa}, Scotch \cite{Scotch} produces better initial partitions than Metis, and therefore we also use it in \algname. 

\paragraph*{Refinement}\label{p:refinement}
After a matching is uncontracted during the refinement phase, some local improvement methods are applied in order to reduce the cut while maintaining the balancing constraint.

We implemented two kinds of local improvement schemes within our framework. 
The first scheme is so called \emph{quotient graph style refinement} \cite{kappa}. 
This approach uses the underlying \emph{quotient graph}. 
Each edge in the quotient graph yields a pair of blocks which share a non empty boundary. 
On each of these pairs we can apply a two-way local improvement method which only moves nodes between 
the current two blocks. 
Note that this approach enables us to integrate flow based improvement techniques between two blocks which are described in Section \ref{s:flowrefinement}.
\begin{figure}[b]
\begin{center}
\begin{tabular}{ccc}
\includegraphics[width=80pt]{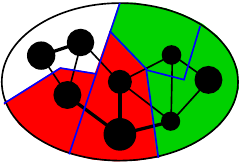}&\, \, \, \, \, \, \, \,&\includegraphics[width=80pt]{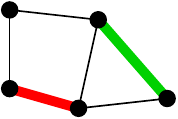}
\end{tabular}
\end{center}
\caption{A graph which is partitioned into five blocks and its corresponding
          quotient graph $\mathcal{Q}$ which has five nodes and six edges. Two pairs of blocks are highlighted in red and green.}
        \label{fig:pictureoverview}
\end{figure}

Our two-way local search algorithm works as in KaPPa \cite{kappa}. We present it here for completeness.
It is basically the FM-algorithm \cite{fiduccia1982lth}: For each of the two blocks $A$, $B$ under consideration,
a priority queue of nodes eligible to move is kept. The priority is based on
the \emph{gain}, i.e., the decrease in edge cut when the node is moved to the
other side. Each node is moved at most once within a single local search. The
queues are initialized in random order with the nodes at the partition
boundary. 

There are different possibilities to select a block from which a node shall be moved. 
The classical FM-algorithm \cite{fiduccia1982lth} alternates between both blocks. 
We employ the \emph{TopGain} strategy from \cite{kappa} which selects the block 
with the largest gain and breaks ties randomly if the the gain values are equal. 
In order to achieve a good balance, TopGain adopts the exception that the block with larger weight is used when one
of the blocks is overloaded.  After a stopping criterion is applied we rollback to the best found cut within the balance constraint.

The second scheme is so call \emph{$k$-way local search}. This method has a more global view since it is not restricted to moving nodes between two blocks only. It also basically the FM-algorithm \cite{fiduccia1982lth}. We now outline the variant we use. 
Our variant uses only one priority queue $P$ which is initialized with a subset $S$ of the partition boundary in a random order. 
The priority is based on the max gain $g(v) = \max_P g_P(v)$ where $g_P(v)$ is the decrease in edge cut when moving $v$ to block $P$.  
Again each node is moved at most once. 
Ties are broken randomly if there is more than one block that will give max gain when moving $v$ to it. 
Local search then repeatedly looks for the highest gain node $v$. 
However a node $v$ is not moved, if the movement would lead to an unbalanced partition. 
The $k$-way local search is stopped if the priority queue $P$ is empty (i.e. each node was moved once) or a stopping criteria described below applies. Afterwards the local search is rolled back the lowest cut fulfilling the balance condition that occurred during this local search. 
This procedure is then repeated until no improvement is found or a maximum number of iterations is reached.

We adopt the stopping criteria proposed in KaSPar \cite{kaspar}. 
This stopping rule is derived using a random walk model. 
Gain values in each step are modelled as identically distributed, independent random variables whose expectation $\mu$ and variance $\sigma^2$ is obtained from the previously observed $p$ steps since the last improvement. 
Osipov and Sanders \cite{kaspar} derived that it is unlikely for the local search to produce a better cut if 
$$p\mu^2 > \alpha \sigma^2 + \beta$$
for some tuning parameters $\alpha$ and $\beta$. The Parameter $\beta$ is a base value that avoids stopping just after a small constant number of steps that happen to have small variance. We also set it to $\ln n$.

There are different ways to initialize the queue $P$, e.g. the complete partition boundary or only the nodes which are incident to more than two partitions (corner nodes). 
Our implementation takes the complete partition boundary for initialization.  
In Section~\ref{ss:multitry} we introduce multi-try $k$-way searches which is a more localized $k$-way search inspired by KaSPar \cite{kaspar}. 
This method initializes the priority queue with only a single boundary node and its neighbors that are also boundary nodes.

The main difference of our implementation to KaSPar is that we use only one priority queue. 
KaSPar maintains a priority queue for each block. 
A priority queue is called eligible if the highest gain node in this queue can be moved to its target block without violating the balance constraint. 
Their local search repeatedly looks for the highest gain node $v$ in any eligible priority queue and moves this node.

\section{Related Work}\label{s:related}
There has been a huge amount of research on graph
partitioning so that we refer the reader to \cite{fjallstrom1998agp,SchKarKum00,Walshaw07} for more material.
All general purpose methods that are able to obtain good partitions for large real world
graphs are based on the multilevel principle outlined in
Section~\ref{s:preliminaries}. The basic idea can be traced back to multigrid
solvers for solving systems of linear equations \cite{Sou35,Fedorenko61} but
more recent practical methods are based on mostly graph theoretic aspects in
particular edge contraction and local search.  Well known software packages
based on this approach include Chaco \cite{Chaco}, Jostle~\cite{Walshaw07},
Metis \cite{SchKarKum00}, Party \cite{Party}, and Scotch \cite{Scotch}.  

KaSPar \cite{kaspar} is a new graph partitioner based on the central idea to (un)contract only a single edge between two levels. 
It previously obtained the best results for many of the biggest graphs in \cite{Walshawbench}. 

KaPPa \cite{kappa} is a "classical" matching based MGP algorithm designed for scalable parallel execution and its local search only considers independent pairs of blocks at a time. 

DiBaP \cite{meyerhenke2008ndb} is a multi-level graph partitioning package
where local improvement is based on diffusion which also yields partitions of very high quality. 

MQI \cite{lang2004flow} and Improve \cite{andersen2008algorithm} are flow-based methods for improving graph cuts when cut quality is measured by quotient-style metrics such as expansion or conductance. Given an undirected graph with an initial partitioning, they build up a completely new directed graph which is then used
to solve a max flow problem.  Furthermore, they have been able to show that there is an improved quotient cut if and only if the maximum flow is less than $ca$, where $c$ is the initial cut and $a$ is the number of vertices in the smaller block of the initial partitioning. This approach is
currently only feasible for $k=2$.
Improve also uses several minimum cut computations to improve the quotient cut score of a proposed partition. Improve always beats or ties MQI.

Very recently an algorithm called PUNCH \cite{delling2010graph} has been introduced. 
This approach is not based on the multilevel principle.  
However, it creates a coarse version of the graph based on the notion of natural cuts. 
Natural cuts are relatively sparse cuts close to denser areas. 
They are discovered by finding minimum cuts between carefully chosen regions of the graph. 
Experiments indicate that the algorithm computes very good cuts for road networks.  
For instances that don't have a natural structure such as road networks, natural cuts are not very helpful.

The concept of \emph{iterated multilevel algorithms} was introduced by \cite{toulouse1999multi,walshaw2004multilevel}. The main idea is to iterate the coarsening and uncoarsening phase and use the information gathered. That means that once the graph is partitioned, edges that are between
two blocks will not be matched and therefore will also not be contracted. This ensures increased quality of the partition if the refinement algorithms guarantees not to find a worse partition than the initial one. 

\section{Local Improvement}\label{s:refinement}
Recall that once a matching is uncontracted a local improvement method tries to reduce the cut size of the projected partition.
We now present two novel local improvement methods.
The first method which is described in Section \ref{s:flowrefinement} is based on  max-flow min-cut computations 
between pairs of 
\begin{figure}[t]
\begin{center}
\includegraphics[width=0.4\textwidth]{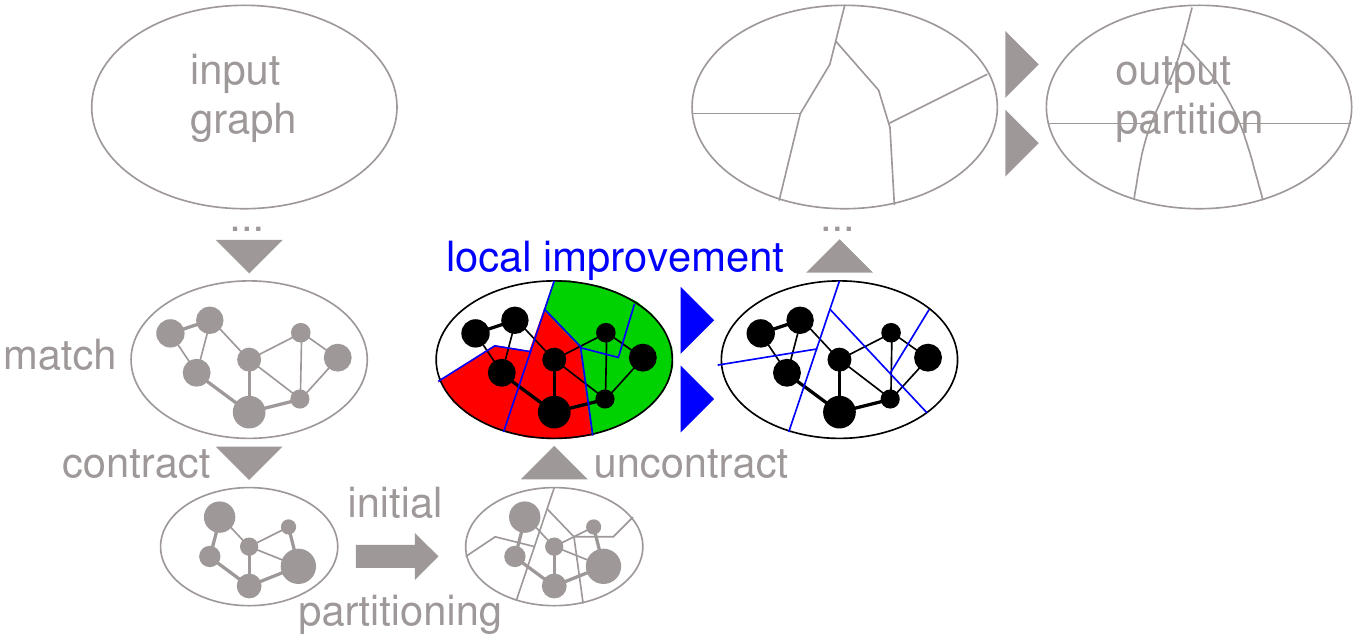}
\end{center}
\caption{After a matching is uncontracted a local improvement method is applied. }
\end{figure}
 blocks, i.e. improving a given 2-partition. 
 Since each edge of the quotient graph yields a pair of blocks which share a non empty boundary,  we integrated this method into the quotient graph style refinement scheme which is described in 
Section~\ref{p:refinement}.  The second method which is described in Section \ref{ss:multitry} is called multi-try FM which is a more localized $k$-way local search.
Roughly speaking, a $k$-way local search is repeatedly started with a priority queue which is initialized with only one random boundary node and its neighbors that are also boundary nodes. 
At the end of the section we shortly show how the pairwise refinements can be scheduled and how the more localized search can be incorporated with this scheduling.  
\subsection{Using Max-Flow Min-Cut Computations for Local Improvement}\label{s:flowrefinement}
We now explain how flows can be used to improve a given partition of two blocks and therefore can be used as a refinement algorithm in a multilevel framework. For simplicity we assume $k=2$. However it is clear that this refinement method fits perfectly into the quotient graph style refinement
algorithms. 

To start with the description of the constructed max-flow min-cut problem, we need a few notations. 
Given a two-way partition $P: V \to \{1,2\}$ of a graph $G$ we define the \emph{boundary nodes} as $\delta := \{u \mid \exists (u,v) \in E : P(u) \neq P(v)\}$. 
We define \emph{left boundary nodes} to be $\delta_l := \delta \cap \{u \mid P(u) = 1\}$ and \emph{right boundary nodes} to be $\delta_r := \delta \cap \{u \mid P(u) = 2\}$.
Given a set of nodes $B \subset V$ we define its \emph{border} $\partial B := \{ u \in B \mid \exists (u,v) \in E : v \not\in B\}$. Unless otherwise mentioned we call $B$ \emph{corridor} because it will be a zone around the initial cut. 
The set $\partial_l B := \partial B \cap \{u \mid P(u) = 1\}$ is called \emph{left corridor border} and the set $\partial_r B := \partial B \cap \{ u \mid P(u) = 2\}$ is called \emph{right corridor border}. 
We say an \emph{$B$-corridor induced subgraph} $G'$ is the node induced subgraph $G[B]$ plus two nodes $s,t$ and additional edges starting from $s$ or edges ending in $t$.
An $B$-corridor induced subgraph has the \emph{cut property $C$} if each ($s$,$t$)-min-cut in $G'$ induces a cut within the balance constrained in $G$.

The main idea is to construct a $B$-corridor induced subgraph $G'$ with cut property $C$. 
On this graph we solve the max-flow min-cut problem. The computed min-cut yields a feasible improved cut within the balance constrained in $G$.  
The construction is as follows (see also Figure \ref{fig:flowconstruction}). 
\begin{figure}[h]
\begin{center}
\begin{tabular}{c}
\includegraphics[width=140pt]{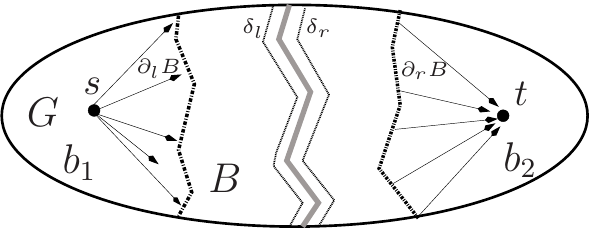}\, \, \includegraphics[width=140pt]{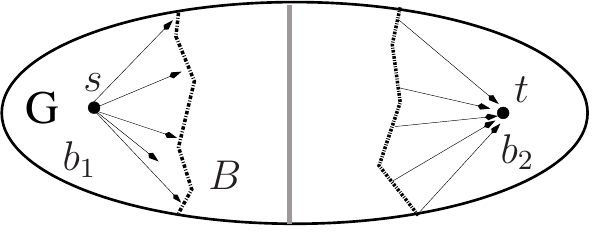}
\end{tabular}
\end{center}
\caption{ The construction of a feasible flow problem which yields optimal cuts in $G'$ and an improved cut within the balance constraint in $G$. On the top the initial construction is shown and on the bottom we see the improved partition.}
        \label{fig:flowconstruction}
\end{figure}

First we need to find a corridor $B$ such that the $B$-corridor induced subgraph will have the cut property $C$. 
This can be done by performing two Breadth First Searches (BFS). Each node touched during these searches belongs to the corridor $B$.
The first BFS is initialized with the left boundary nodes $\delta_l$. It is only expanded with nodes that are in block 1. 
As soon as the weight of the area found by this BFS would exceed $(1+\epsilon) c(V)/2 - w(\text{block 2})$, we stop the BFS. 
The second BFS is done for block 2 in an analogous fashion.  

In order to achieve the cut property $C$, the $B$-corridor induced subgraph $G'$ gets additional $s$-$t$ edges. 
More precisely $s$ is connected to all left corridor border nodes $\partial_l B$ and all right corridor border nodes $\partial_r B$ are connected to $t$. All of these new edges get the edge weight $\infty$. Note that this are directed edges.

The constructed $B$-corridor subgraph $G'$ has the cut property $C$ since the worst case new weight of block 2 is lower or equal to $w(\text{block 2}) + (1+\epsilon)c(V)/2 - w(\text{block 2}) = (1+\epsilon)c(V)/2$. Indeed the same holds for the worst case new weight of block 1. 

There are multiple ways to improve this method. First, if we found an improved edge cut, we can apply this method again since the initial boundary has changed which implies that it is most likely that the corridor $B$ will also change. 
Second, we can adaptively control the size of the corridor $B$  which is found by the BFS. 
This enables us to search for cuts that fulfill our balance constrained even in a larger corridor ( say $\epsilon'=\alpha\epsilon$ for some parameter $\alpha$ ), i.e. if the found min-cut in $G'$ for $\epsilon'$ fulfills the balance constraint in $G$, we accept it and increase $\alpha$ to $\min(2\alpha, \alpha')$ where $\alpha'$ is an upper bound for $\alpha$. 
Otherwise the cut is not accepted and we decrease $\alpha$ to $\max(\frac{\alpha}{2},1)$. 
This method is iterated until a maximal number of iterations is reached or if the computed cut yields a feasible partition without an decreased edge cut. We call this method \textit{adaptive flow iterations}. 

\subsubsection{Most Balanced Minimum Cuts}
Picard and Queyranne have been able to show that one $(s,t)$ max-flow contains information about all minimum ($s$,$t$)-cuts in the graph. 
Here finding all minimum cuts reduces to a straight forward enumeration. 
Having this in mind the idea to search for min-cuts in larger corridors becomes even more attractive. 
Roughly speaking, we present a heuristic that, given a max-flow, creates min-cuts that are better balanced. 
First we need a few notations. For a graph $G=(V,E)$ a set $C \subseteq V$ is a \textit{closed vertex set} iff for all vertices $u,v \in V$, the conditions $u \in C$ and $(u,v) \in E$ imply $v \in C$. An example can be found in Figure~\ref{fig:closedvertexset}.
\begin{lemma}[Picard and Queyranne \cite{picard1980structure}]
There is a 1-1 correspondence between the minimum $(s,t)$-cuts of a graph and the closed vertex sets containing $s$ in the residual graph of a maximum $(s,t)$-flow.  
\end{lemma}

To be more precise for a given closed vertex set $C$ containing $s$ of the residual graph the corresponding min-cut is $(C, V \backslash C)$. Note that distinct maximum flows may produce different residual graphs but the set of closed vertex sets remains the same. To enumerate all minimum cuts of a graph \cite{picard1980structure} a further reduced graph is computed which is described below. 
However, the problem of finding the minimum cut with the best balance (most balanced minimum cut) is NP-hard \cite{feige2006finding,bonsma2010most}.  
\begin{figure}[t]
\begin{center}
\includegraphics[width=100pt]{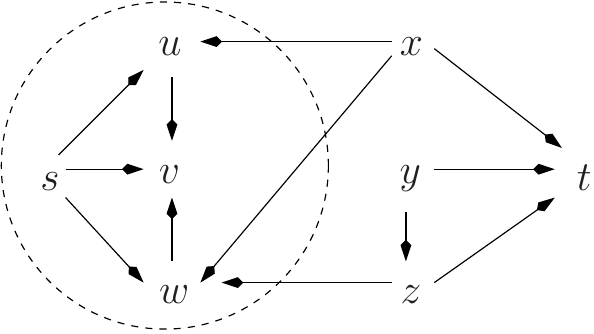}
\end{center}
\caption{A small graph where $C=\{s,u,v,w\}$ is a closed vertex set.}
\label{fig:closedvertexset}
\end{figure}

The minimum cut that is identified by the labeling procedure of Ford and Fulkerson \cite{ford1962flows} is the one with the smallest possible source set. 
We now define how the representation of the residual graph can be made more compact \cite{picard1980structure} and then explain the heuristic we use to obtain closed vertex sets on this graph to find min-cuts that have a better balance. 
After computing a maximum $(s,t)$-flow, we compute the strongly connected components of the residual graph using the algorithm proposed in \cite{cheriyan1996algorithms,gabow2000path}.
We make the representation more compact by contracting these components and refer to it as \textit{minimum cut representation}. 
This reduction is possible since two vertices that lie on a cycle have to be in the same closed vertex set of the residual graph.
The result is a weighted, directed and acyclic graph (DAG). Note that each closed vertex set of the minimum cut representation induces a minimum cut as well. 

As proposed in \cite{picard1980structure} we make the minimum cut representation even more compact:
We eliminate the component $T$ containing the sink $t$, and all its predecessors (since they cannot belong to a closed vertex set not containing $T$) and the component $S$ containing the source, and all its successors (since they must belong to a closed vertex set containing $S$) using a BFS. 

We are now left with a further reduced graph. 
On this graph we search for closed vertex sets (containing $S$) since they still induce $(s,t)$-min-cuts in the original graph. This is done by using the following heuristic which is repeated a few times.  
The main idea is that a topological order yields complements of closed vertex sets quite easily. 
Therefore, we first compute a random topological order, e.g. using a randomized DFS. 
Next we sweep through this topological order and sequentially add the components to the complement of the closed vertex set.
Note that each of the computed complements of closed vertex sets $\tilde{C}$ also yields a closed vertex set ($V\backslash \tilde{C}$).
That means by sweeping through the topological order we compute closed vertex sets each inducing a min-cut having a different balance.  
We stop when we have reached the best balanced minimum cut induced through this topological order with respect to the original graph partitioning problem.
The closed vertex set with the best balance occurred during the repetitions of this heuristic is returned. 
Note in large corridors this procedure may finds cuts that are not feasible, e.g. if there is no feasible minimum cut. 
Therefore the algorithm is combined with the adaptive strategy from above. 
We call this method \textit{balanced adaptive flow iterations}.
\begin{figure}[t!]
\begin{center}
\begin{tabular}{c}
\includegraphics[width=140pt]{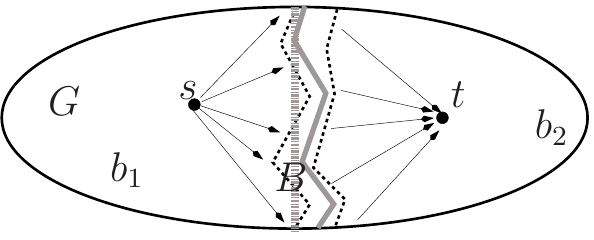} \, \, \includegraphics[width=140pt]{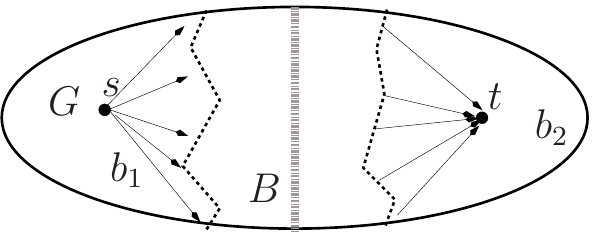}
\end{tabular}
\end{center}
\caption{In the situation on the top it is not possible in the small corridor around the initial cut to find the dashed minimum cut which has optimal balance; however if we solve a larger flow problem on the bottom and search for a cut with good balance we can find the dashed minimum cut with optimal
        balance but not every min cut is feasible for the underlying graph partitioning problem. }
\label{fig:bestbalancedcutuseful}
\end{figure}
\pagebreak
\subsection{Multi-try FM} \label{ss:multitry}
This refinement variant is organized in rounds. In each round we put \emph{all} boundary nodes of the current block pair into a todo list. The todo list is then permuted. 
Subsequently, we begin a $k$-way local search starting with a random node of this list if it is still a boundary node and its neighboring nodes that are also boundary nodes. 
Note that the difference to the global $k$-way search described in Section~\ref{p:refinement} is the initialisation of the priority queue. 
If the selected random node was already touched by a previous $k$-way search in this round then no search is started. 
Either way, the node is removed from the todo list (simply swapping it with the last element and executing a pop\_back on that list).
For a $k$-way search it is not allowed to move nodes that have been touched in a previous run.
This way we can assure that at most $n$ nodes are touched during one round of the algorithm. 
This algorithm uses the adaptive stopping criteria from KaSPar which is described in Section~\ref{p:refinement}.  

\subsection{Scheduling Quotient Graph Refinement}
\label{sec:schedulingqgraphedges}
There a two possibilities to schedule the execution of two way refinement algorithms on the quotient graph. 
Clearly the first simple idea is to traverses the edges of $Q$ in a random order and perform refinement on them. This is iterated until no change occurred or a maximum number of iterations is reached. 
The second algorithm is called \textit{active block scheduling}.
The main idea behind this algorithm is that the local search should be done in areas in which change still happens and therefore avoid unnecessary local search. 
The algorithm begins by setting every block of the partition \textit{active}. 
Now the scheduling is organized in rounds. 
In each round, the algorithm refines adjacent pairs of blocks, which have at least one active block, in a random order. 
If changes occur during this search both blocks are marked active for the next round of the algorithm. 
After each pair-wise improvement a multi-try FM search ($k$-way) is started. 
It is initialized with the boundaries of the current pair of blocks. Now each block which changed during this search is also marked active. 
The algorithm stops if no active block is left.  
Pseudocode for the algorithm can be found in the appendix in Figure~\ref{fig:activeblockscheduling}.

\section{Global Search}\label{s:globalsearch}
Iterated Multilevel Algorithms where introduced by \cite{toulouse1999multi,walshaw2004multilevel} (see Section \ref{s:related}).
For the rest of this paper Iterated Multilevel Algorithms are called $V$-cycles unless otherwise mentioned. 
The main idea is that if a partition of the graph is available then it can be reused during the coarsening and uncoarsening phase.
To be more precise, the multi-level scheme is repeated several times and once the graph is partitioned, edges between two blocks will not be matched and therefore will also not be contracted such that a given partition can be used as initial partition of the coarsest graph.  
This ensures increased quality of the partition if the refinement algorithms guarantees not to find a worse partition than the initial one. 
Indeed this is only possible if the matching includes non-deterministic factors such as random tie-breaking, so that each iteration is very likely to give different coarser graphs. 
Interestingly, in multigrid linear solvers Full-Multigrid methods are generally preferable to simple $V$-cycles \cite{briggs2000multigrid}.
Therefore, we now introduce two novel global search strategies namely \textit{W-cycles} and \textit{F-cycles} for graph partitioning. 
A W-cycle works as follows: on \emph{each} level we perform \emph{two independent trials} using different random seeds for tie breaking during contraction, and local search. 
As soon as the graph is partitioned, edges that are between blocks are not matched.  
A F-cycle works similar to a W-cycle with the difference that the global number of independent trials on each level is bounded by 2. 
Examples for the different cycle types can be found in Figure~\ref{fig:cycles} and Pseudocode can be found in Figure~\ref{fig:globalsearchpseudocode}.
Again once the graph is partitioned for the first time, then this partition is used in the sense that edges between two blocks are not contracted. 
In most cases the initial partitioner is not able to improve this partition from scratch or even to find this partition. 
Therefore no further initial partitioning is used if the graph already has a partition available. 
These methods can be used to find very high quality partitions but on the other hand they are more expensive than a single MGP run. 
However, experiments in Section~\ref{s:experiments} show that all cycle variants are more efficient than simple plain restarts of the algorithm.
In order to bound the runtime we introduce a level split parameter $d$ such that the independent trials are only performed every $d$'th level. We go into more detail after we have analysed the run time of the global search strategies.  

\begin{figure}[h]
\begin{center}
\includegraphics[width=30pt]{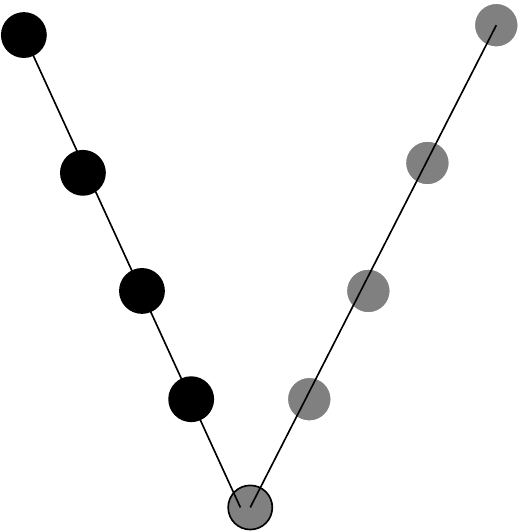} 
\includegraphics[width=107pt]{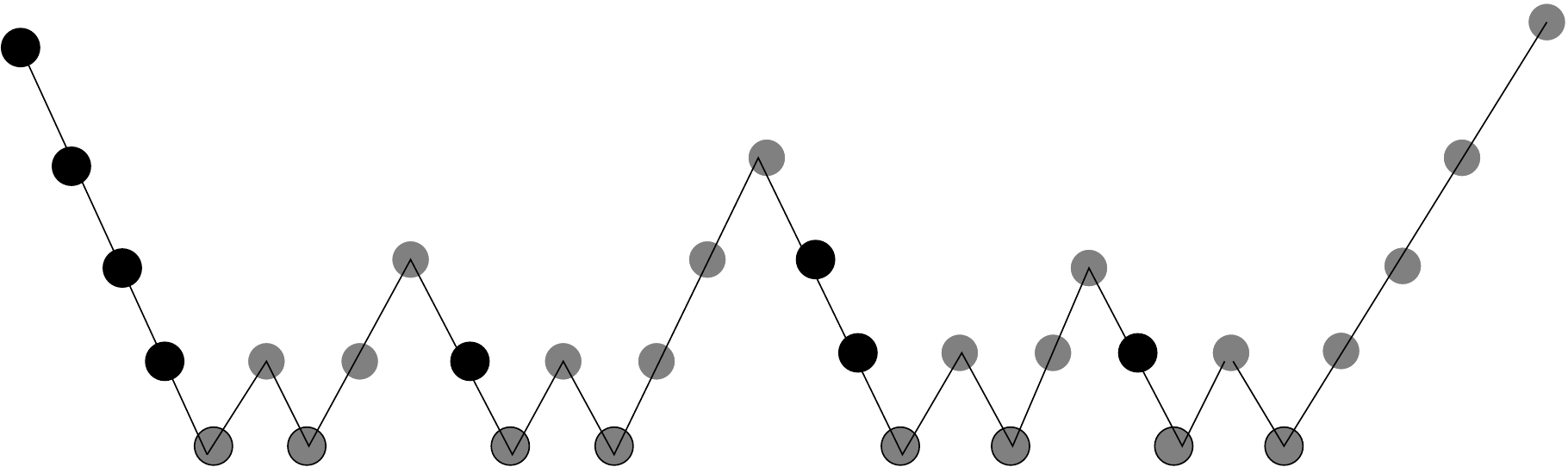}  
\includegraphics[width=82pt]{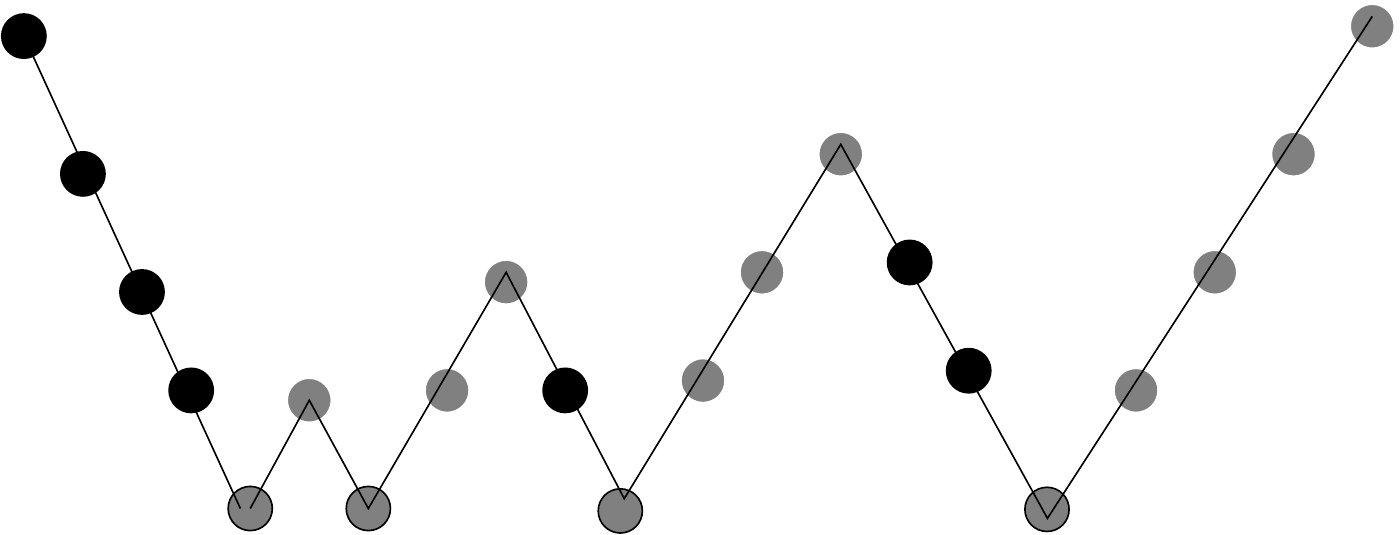}
\end{center}
\caption{From left to right: A single MGP V-cycle, a W-cycle and a F-cycle.} 
\label{fig:cycles}
\end{figure}
\vfill
\newpage 
\paragraph*{Analysis}
We now roughly analyse the run time of the different global search strategies under a few assumptions. In the following the shrink factor names the factor the graph shrinks during one coarsening step. 
\begin{theorem} 
If the time for coarsening and refinement is $T_{\text{cr}}(n) := bn$ and a constant shrink factor $a \in [1/2,1)$ is given. Then:
\begin{equation}
   T_{W,d}(n) \begin{cases}   
                        \lessapprox \frac{1-a^d}{1-2a^d}T_V(n)& \text{if } 2a^d < 1 \\
                        \in \Theta(n \log n) & \text{if } 2a^d = 1\\
                \in \Theta(n^{\frac{\log 2}{\log{\frac{1}{a^d}}}}) & \text{if } 2a^d > 1 \\
              \end{cases}
\end{equation}

\begin{equation}
   T_{F,d}(n) \leq \frac{1}{1-a^d}T_V(n)
\end{equation}
where $T_V$ is the time for a single V-cycle and $T_{W,d}$,$T_{F,d}$ are the time for a W-cycle and F-cycle with level split parameter $d$.  
\end{theorem}
\begin{proof} 
The run time of a single V-cycle is given by $T_V(n) = \sum_{i=0}^l T_{\text{cr}}(a^in) = bn \sum_{i=0}^{l} a^i = bn(1-a^{l+1})/(1-a)$.
The run time of a W-cycle with level split parameter $d$ is given by the time of $d$ coarsening and refinement steps plus the time of the two trials on the created coarse graph. For the case $2a^d  < 1$ we get  
\begin{eqnarray*}
T_{W,d}(n) &=& bn\sum_{i=0}^{d-1}a^i + 2T_{W,d}(a^dn) \leq bn \frac{1-a^d}{1-a} \sum_{i=0}^{\infty} (2a^d)^i \\ 
           &\leq&  \frac{1-a^d}{(1-a^{l+1})(1-2a^d)}T_V(n) \approx \frac{1-a^d}{1-2a^d}T_V(n).
\end{eqnarray*}

The other two cases for the W-cycle follow directly from the master theorem for analyzing divide-and-conquer recurrences. 
To analyse the run time of a F-cycle we observe that 
\begin{eqnarray*}
T_{F,d}(n) & \leq & \sum_{i=0}^{l} T_{\text{cr}}(a^{i\cdot d}n) \leq \frac{bn}{1-a} \sum_{i=0}^{\infty} (a^d)^i =  \frac{1}{1-a^d}T_V(n)
\end{eqnarray*}
where $l$ is the total number of levels. This completes the proof of the theorem.
\end{proof}
Note that if we make the optimistic assumption that $a=1/2$ and set $d=1$ then a F-cycle is only twice as expensive as a single V-cycle. 
If we use the same parameters for a W-cycle we get a factor $\log n$ asymptotic larger execution times. 
However in practice the shrink factor is usually worse than $1/2$. 
That yields an even larger asymptotic run time for the W-cycle (since for $d=1$ we have $2a>1$).
Therefore, in order to bound the run time of the W-cycle the choice of the level split parameter $d$ is crucial. 
Our default value for $d$ for W- and F-cycles is 2, i.e. independent trials are only performed every second level. 
\newpage
\section{Experiments}\label{s:experiments}
\paragraph*{Implementation}
We have implemented the algorithm described above using C++.  Overall,
our program consists of about 12\,500 lines of code. 
Priority queues for the local search are based on binary heaps. 
Hash tables use the library (extended STL) provided with the GCC compiler.
For the following comparisons we used Scotch 5.1.9., DiBaP 2.0.229 and kMetis 5.0 (pre2).
The flow problems are solved using Andrew Goldbergs Network Optimization Library HIPR \cite{cherkassky1997implementing} which is integrated into our code.
\paragraph*{System}
We have run our code on a cluster where each node is equipped with two Quad-core
Intel Xeon processors (X5355) which run at a clock speed of 2.667 GHz, has 2x4
MB of level 2 cache each and run Suse Linux Enterprise 10 SP 1. Our program was compiled using GCC 
Version 4.3.2 and optimization level~3. 
\paragraph*{Instances}
We report experiments on two suites of instances summarized in
the appendix in Table~\ref{tab:instances}. These are the same instances as used for the evaluation of KaPPa \cite{kappa}. 
We present them here for completeness. 
\Id{rggX} is a \emph{random geometric graph} with
$2^{X}$ nodes where nodes represent random points in the unit square and edges
connect nodes whose Euclidean distance is below $0.55 \sqrt{ \ln n / n }$.
This threshold was chosen in order to ensure that the graph is almost connected. 
\Id{DelaunayX} is the Delaunay triangulation of $2^{X}$
random points in the unit square.  Graphs \Id{bcsstk29}..\Id{fetooth} and
\Id{ferotor}..\Id{auto} come from Chris Walshaw's benchmark archive
\cite{walshaw2000mpm}.  Graphs \Id{bel}, \Id{nld}, \Id{deu} and \Id{eur} are
undirected versions of the road networks of Belgium, the Netherlands, Germany,
and Western Europe respectively, used in \cite{DSSW09}. Instances
\Id{af\_shell9} and \Id{af\_shell10} come from the Florida Sparse Matrix Collection \cite{UFsparsematrixcollection}.  
For the number of partitions $k$ we choose the values
used in  \cite{walshaw2000mpm}: 2, 4, 8, 16, 32, 64.
Our default value for the allowed imbalance is 3 \% since this is one
of the values used in \cite{walshaw2000mpm} and the default value in Metis.

\paragraph*{Configuring the Algorithm}

We currently define three configurations of our algorithm: Strong, Eco and Fast.
The configurations are described below.

\textbf{KaFFPa Strong:}
The aim of this configuration is to obtain a graph partitioner that is able to achieve the best known partitions for many standard benchmark instances.
It uses the GPA algorithm as a matching algorithm combined with the rating function $\expansion^{*2}$.
However, the rating function $\expansion^{*2}$ has the disadvantage that it evaluates to one on the first level of an unweighted graph. 
Therefore, we employ $\innerOuter$ on the first level to infer structural information of the graph. 
We perform $100/\log k$ initial partitioning attempts using Scotch as an initial partitioner.  
The \textit{refinement phase} first employs $k$-way refinement (since it converges very fast) which is initialized with the complete partition boundary. 
It uses the adaptive search strategy from KaSPar \cite{kaspar} with $\alpha = 10$. 
The number of rounds is bounded by ten. 
However, the $k$-way local search is stopped as soon as a $k$-way local search round did not find an improvement.
We continue by performing quotient-graph style refinement. 
Here we use the active block scheduling algorithm which is combined with the multi-try local search (again $\alpha = 10$) as described in Section \ref{sec:schedulingqgraphedges}.
A pair of blocks is refined as follows: We start with a pairwise FM search which is followed by the max-flow min-cut algorithm (including the most balancing cut heuristic).  
The FM search is stopped if more than 5\% of the number of nodes in the current block pair have been moved without yielding an improvement. 
The upper bound factor for the flow region size is set to $\alpha'=8$.
As \textit{global search strategy} we use two F-cycles. 
Initial Partitioning is only performed if previous partitioning information is \emph{not} available. 
Otherwise, we use the given input partition. 

\textbf{KaFFPa Eco:} The aim of KaFFPa Eco is to obtain a graph partitioner that is fast on the one hand and on the other hand is able to compute partitions of high quality.
This configuration matches the first $\max(2,7-\log k)$ levels using a random matching algorithm. 
The remaining levels are matched using the GPA algorithm employing the edge rating function $\expansion^{*2}$. 
It then performs $\min(10,40/\log k)$ initial partitioning repetitions using Scotch as initial partitioner. 
The refinement is configured as follows: again we start with $k$-way refinement as in KaFFPa-Strong. 
However, for this configuration the number of $k$-way rounds is bounded by $\min(5, \log k)$.
We then apply quotient-graph style refinements as in KaFFPa Strong; again with slightly different parameters. 
The two-way FM search is stopped if 1\% of the number of nodes in the current block pair has been moved without yielding an improvement. 
The flow region upper bound factor is set to $\alpha'=2$.
We do not apply a more sophisticated global search strategy in order to be competitive regarding runtime. 

\textbf{KaFFPa Fast:} The aim of KaFFPa Fast is to get the fastest available system for large graphs while still improving partitioning quality to the previous fastest system. 
KaFFPa Fast matches the first four levels using a random matching algorithm. It then continues by using the GPA algorithm equipped with $\expansion^{*2}$ as a rating function. 
We perform exactly one initial partitioning attempt using Scotch as initial partitioner. 
The refinement phase works as follows: for $k\leq 8$ we only perform quotient-graph refinement: each pair of blocks is refined exactly once using the pair-wise FM algorithm. 
Pairs of blocks are scheduled randomly. 
For $k>8$ we only perform one $k$-way refinement round. 
In both cases the local search is stopped as soon as 15 steps have been performed without yielding an improvement. 
Note that using flow based algorithms for refinement is already too expensive.
Again we do not apply a more sophisticated global search strategy in order to be competitive regarding runtime.

\paragraph*{Experiment Description}
We performed two types of experiments namely normal tests and tests for effectiveness. Both are described below. 

\textbf{Normal Tests:} Here we perform 10 repetitions for the small networks and 5 repetitions for the other. We report the arithmetic average of computed cut size, running time and the best cut found. 
When further averaging over multiple instances, we use the geometric mean in order to give every instance the same influence on the \textit{final score}.  
\footnote{\AuthorChris{Because we have multiple repetitions for each instance (graph, $k$), we compute the geometric mean of the average (\textbf{Avg.}) edge cut values for each
instance or the geometric mean of the best (\textbf{Best.}) edge cut value occurred. The same is done for the runtime \textbf{t} of each algorithm configuration.}}  

\textbf{Effectiveness Tests:} Here each algorithm configuration has the same time for computing a partition. 
Therefore, for each graph and $k$ each configuration is executed once and we remember the largest execution time $t$ that occurred.  
Now each algorithm gets time $3t$ to compute a good partition, i.e. taking the best partition out of repeated runs. If a variant can perform a next run depends on the remaining time, i.e. we flip a coin with corresponding probabilities such that the expected time over multiple runs is $3t$.  This is repeated 5 times. The final score is computed as in the normal test using these values.
\\

\subsection{Insights about Flows}
We now evaluate how much the usage of max-flow min-cut algorithms improves the final partitioning results and check its effectiveness. 
For this test we use a basic two-way FM configuration to compare with. 
This basic configuration is modified as described below to look at a specific algorithmic component regarding flows.  
It uses the Global Paths Algorithm as a matching algorithm and performs five initial partitioning attempts using Scotch as initial partitioner. 
It further employs the active block scheduling algorithm equipped with the two-way FM algorithm described in Section~\ref{p:refinement}. 
The FM algorithm stopps as soon as 5\% of the number of nodes in the current block pair have been moved without yielding an improvement. 
Edge rating functions are used as in KaFFPa Strong.
Note that during this test our main focus is the evaluation of flows and therefore we don't use $k$-way refinement or multi-try FM search. 
For comparisons this basic configuration is extended by specific algorithms, e.g. a configuration that uses Flow, FM and the most balanced cut heuristics (MB). This configuration is then indicated by (+Flow, +FM, +MB). 
\begin{table}[b!]
\small
\begin{center}
\hspace*{-0.8cm}        
\begin{tabular}{|l|r|r|r|r|r|r|r|r|r|r|r|r|r|r|r|r|}\hline
                        Variant                            & \multicolumn{4}{|c|}{(+Flow, -MB, -FM )}                                                                                                                                 & \multicolumn{4}{|c|}{(+Flow, +MB, -FM)}                                           & \multicolumn{4}{|c|}{(+Flow, -MB, +FM)}                               & \multicolumn{4}{|c|}{(+Flow, +MB, +FM)}                             \\
                        \hline
                         $\alpha' $ & Avg.                                     & Best.                                   & Bal.                                    & $t$                                       & Avg.                      & Best.             & Bal.            & $t$             & Avg.            & Best.           & Bal.            & $t$             & Avg.            & Best.           & Bal.            & $t$ \\
                        \hline
                        $16$                               & \numprint{-1.88}                         & \numprint{-1.28}                        & \numprint{1.03}                         & \numprint{4.17}                           & \numprint{0.81}           & \numprint{0.35}   & \numprint{1.02} & \numprint{3.92} & \numprint{6.14} & \numprint{5.44} & \numprint{1.03} & \numprint{4.30} & \numprint{7.21} & \numprint{6.06} & \numprint{1.02} & \numprint{5.01}\\
                        $8$                                & \numprint{-2.30}                         & \numprint{-1.86}                        & \numprint{1.03}                         & \numprint{2.11}                           & \numprint{0.41}           & \numprint{-0.14}  & \numprint{1.02} & \numprint{2.07} & \numprint{5.99} & \numprint{5.40} & \numprint{1.03} & \numprint{2.41} & \numprint{7.06} & \numprint{5.87} & \numprint{1.02} & \numprint{2.72}\\
                        $4$                                & \numprint{-4.86}                         & \numprint{-3.78}                        & \numprint{1.02}                         & \numprint{1.24}                           & \numprint{-2.20}          & \numprint{-2.80}  & \numprint{1.02} & \numprint{1.29} & \numprint{5.27} & \numprint{4.70} & \numprint{1.03} & \numprint{1.62} & \numprint{6.21} & \numprint{5.36} & \numprint{1.02} & \numprint{1.76}\\
                        $2$                                & \numprint{-11.86}                        & \numprint{-10.35}                       & \numprint{1.02}                         & \numprint{0.90}                           & \numprint{-9.16}          & \numprint{-8.24}  & \numprint{1.02} & \numprint{0.96} & \numprint{3.66} & \numprint{3.37} & \numprint{1.02} & \numprint{1.31} & \numprint{4.17} & \numprint{3.82} & \numprint{1.02} & \numprint{1.39}\\
                        $1$                                & \numprint{-19.58}                        & \numprint{-18.26}                       & \numprint{1.02}                         & \numprint{0.76}                           & \numprint{-17.09}         & \numprint{-16.39} & \numprint{1.02} & \numprint{0.80} & \numprint{1.64} & \numprint{1.68} & \numprint{1.02} & \numprint{1.19} & \numprint{1.74} & \numprint{1.75} & \numprint{1.02} & \numprint{1.22}\\
                        \hline
                       Ref. & \multicolumn{4}{c|}{(-Flow, -MB, +FM)}                  & \numprint{2974}                          & \numprint{2851}                         & \numprint{1.025}                        & \numprint{1.13}                           & \multicolumn{8}{|c|}{}\\
                        \hline
\end{tabular}

\end{center}

\caption{The final score of different algorithm configurations compared against the basic two-way FM configuration. The parameter $\alpha'$ is the flow region upper bound factor. All average and best cut values except for the basic configuration are improvements relative to the basic configuration in \%.}
\label{tab:flowcomp}
\end{table}

\begin{table}[h!]
\small
\begin{center}
        
\begin{tabular}{|l||r|r||r|r||r|r||}\hline
			Effectiveness  & \multicolumn{2}{c|}{(+Flow, +MB, -FM)} & \multicolumn{2}{c|}{(+Flow,-MB, +FM)} & \multicolumn{2}{c|}{(+Flow,+MB,+FM)}        \\
                         & Avg. & Best. & Avg. & Best. & Avg. & Best. \\
			\hline 
			$\alpha'=1$ & \numprint{-16.41} & \numprint{-16.35}  & \numprint{1.62} & \numprint{1.52}  & \numprint{1.65} & \numprint{1.63}           \\
			2           & \numprint{-8.26}  & \numprint{-8.07}   & \numprint{3.02} & \numprint{2.83}  & \numprint{3.36} & \numprint{3.25}           \\
			4           & \numprint{-3.05}  & \numprint{-3.08}   & \numprint{4.04} & \numprint{3.82}  & \numprint{4.63} & \numprint{4.36}           \\
			8           & \numprint{-1.12}  & \numprint{-1.34}   & \numprint{4.16} & \numprint{4.13}  & \numprint{4.74} & \numprint{4.64}           \\
			16          & \numprint{-1.29}  & \numprint{-1.27}   & \numprint{3.70} & \numprint{3.86}  & \numprint{4.28} & \numprint{4.36}           \\
			\hline                                                                                                                       
                        (-Flow, -MB, +FM)&\numprint{2833}&\numprint{2803}    & \numprint{2831} & \numprint{2801}  & \numprint{2827} & \numprint{2799}           \\
			\hline

\end{tabular}

\end{center}
\caption{Three effectiveness tests each one with six different algorithm configurations. All average and best cut values except for the basic configuration are improvements relative to the basic configuration in \%.}
\label{tab:flowcompeffrel}
\end{table}

In Table~\ref{tab:flowcomp} we see that by Flow on its own, i.e. no FM-algorithm is used at all, we obtain cuts and run times which are worse than the basic two-way FM configuration. 
The results improve in terms of quality and runtime if we enable the most balanced minimum cut heuristic.
Now for $\alpha'=16$ and $\alpha'=8$, we get cuts that are $0.81 \%$ and $0.41 \%$ lower on average than the cuts produced by the basic two-way FM configuration. 
However, these configurations have still a factor four ($\alpha'=16$) or a factor two ($\alpha'=8$) larger run times. 
In some cases, flows and flows with the MB heuristic are not able to produce results that are comparable to the basic two-way FM configuration. 
Perhaps, this is due to the lack of the method to accept suboptimal cuts which yields small flow problems and therefore bad cuts.  
Consequently, we also combined both methods to fix this problem. 
In Table~\ref{tab:flowcomp} we can see that the combination of flows with local search produces up to $6.14 \%$ lower cuts on average than the basic configuration. 
If we enable the most balancing cut heuristic we get on average $7.21 \%$ lower cuts than the basic configuration. 
Since these configurations are the basic two-way FM configuration augmented by flow algorithms they have an increased run time compared to the basic configuration. 
However, Table~\ref{tab:flowcompeffrel} shows that these combinations are also more effective than the repeated execution of the basic two-way FM configuration.
The most effective configuration is the basic two-way FM configuration using flows with $\alpha'=8$ and uses the most balanced cut heuristic. 
It yields $4.73 \%$ lower cuts than the basic configuration in the effectiveness test. 
Absolute values for the test results can be found in Table~\ref{tab:flowcompabs} and Table~\ref{tab:flowcompeff} in the Appendix. 

\subsection{Insights about Global Search Strategies}

In Table~\ref{tab:globalsearchresults} we compared different global search strategies against a single V-cycle. 
This time we choose a relatively fast configuration of the algorithm as basic configuration since the global search strategies are at focus. 
The coarsening phase is the same as in KaFFPa Strong. We perform one initial partitioning attempt using Scotch. 
The refinement employs $k$-way local search followed by quotient graph style refinements. 
Flow algorithms are not enabled for this test. 
The only parameter varied during this test is the global search strategy. 

Clearly, more sophisticated global search strategies decrease the cut but also increase the runtime of the algorithm.
However, the effectiveness results in Table~\ref{tab:globalsearchresults} indicate that repeated executions of more sophisticated global search strategies are always superior to repeated executions of one single V-cycle. 
The largest difference in best cut effectiveness is obtained by repeated executions of $2$ W-cycles and $2$ F-cycles which produce $1.5 \%$ lower best cuts than repeated executions of a normal V-cycle.

The increased effectiveness of more sophisticated global search strategies is due to different reasons. 
First of all by using a given partition in later cycles we obtain a very good initial partitioning for the coarsest graph.
This initial partitioning is usually much better than a partition created by another initial partitioner which
yields good start points for local improvement on each level of refinement.
Furthermore, the increased effectiveness is due to time saved using the active block strategy which converges very quickly in later cycles. 
On the other hand we save time for initial partitioning which is only performed the first time the algorithm arrives in the initial partitioning phase.  

It is interesting to see that although the analysis in Section~\ref{s:globalsearch} makes some simplified assumptions the measured run times in Table~\ref{tab:globalsearchresults} are very close to the values obtained by the analysis.
\begin{table}[h]
\small
\begin{center}
\begin{tabular}{|l|r|r|r|r||r|r|}\hline

                  Algorithm & Avg. & Best & Bal. &$t$ & Eff. Avg. & Eff. Best \\

                  \hline
                  2 F-cycle        & \numprint{2,69} & \numprint{2,45} & \numprint{1.023} & \numprint{2.31} & \numprint{2806} & \numprint{2760}\\
                  3 V-cycle        & \numprint{2,69} & \numprint{2,34} & \numprint{1.023} & \numprint{2.49} & \numprint{2810} & \numprint{2766}\\
                  2 W-cycle        & \numprint{2,91} & \numprint{2,75} & \numprint{1.024} & \numprint{2.77} & \numprint{2810} & \numprint{2760}\\
                  1 W-cycle        & \numprint{1,33} & \numprint{1,10} & \numprint{1.024} & \numprint{1.38} & \numprint{2815} & \numprint{2773}\\
                  1 F-cycle        & \numprint{1,09} & \numprint{1,00} & \numprint{1.024} & \numprint{1.18} & \numprint{2816} & \numprint{2783}\\
                  2 V-cycle        & \numprint{1,88} & \numprint{1,61} & \numprint{1.024} & \numprint{1.67} & \numprint{2817} & \numprint{2778}\\
                  \hline
                  1 V-cycle        & \numprint{2973} & \numprint{2841} & \numprint{1.024} & \numprint{0.85} & \numprint{2834} & \numprint{2801}\\
                  \hline
\end{tabular}
\end{center} 
\caption{Test results for normal and effectiveness tests for different global search strategies. The average cut and best cut values are improvements in \% relative to the basic configuration (1 V-cycle). For F- and W-cycles $d=2$. Absolute values can be found in Table~\ref{tab:globalsearchresultsabs} in the Appendix. }
\label{tab:globalsearchresults}
\end{table}

\vspace*{-1.3cm}
\subsection{Removal / Knockout Tests}\label{ss:parameters}
We now turn into two kinds of experiments to evaluate interactions and relative importance of our algorithmic improvements. 
In the component \textit{removal tests} we take KaFFPa Strong and remove components step by step yielding weaker and weaker variants of the algorithm.  
For the \textit{knockout tests} only one component is removed at a time, i.e. each variant is exactly the same as KaFFPa Strong minus the specified component. 

In the following, \textit{KWay} means the global $k$-way search component of KaFFPa Strong, \textit{Multitry} stands for the more localized $k$-way search during the active block scheduling algorithm  and \textit{-Cyc} means that the F-Cycle component is replaced by one V-cycle. 
Furthermore, \textit{MB} stands for the most balancing minimum cut heuristic, and \textit{Flow} means the flow based improvement algorithms. 

In Table~\ref{tab:removaltestresultscompressed} we see results for the component removal tests and knockout tests.  
More detailed results can be found in the appendix. 
First notice that in order to achieve high quality partitions we don't need to perform classical global $k$-way refinement (KWay).
The changes in solution quality are negligible and both configurations (Strong without KWay and Strong) are equally effective. 
However, the global $k$-way refinement algorithm converges very quickly and therefore speeds up overall runtime of the algorithm; hence we included it into our KaFFPa Strong configuration.

In both tests the largest differences are obtained when the components Flow and/or the Multitry search heuristic are removed.
When we remove all of our new algorithmic components from KaFFPa Strong, i.e global $k$-way search, local multitry search, F-Cycles, and Flow we obtain a graph partitioner that produces 9.3\% larger cuts than KaFFPa Strong.  
Here the effectiveness average cut of the weakest variant in the removal test is about 6.2\% larger than the effectiveness average cut of KaFFPa Strong. 
Also note that as soon as a component is removed from KaFFPa Strong (except for the global $k$-way search) the algorithm gets less effective. 
\begin{table}[h!]
\small
\begin{center}

\begin{tabular}{|r|rrr||rr| } 
\hline
Variant & Avg. & Best. & $t$ & Eff. Avg. & Eff. Best.  \\
\hline
Strong    & \numprint{2683}  & \numprint{2617}  & \numprint{8.93}       & \numprint{2636} & \numprint{2616}\\
\hline
-KWay     & \numprint{-0.04} & \numprint{-0.11} & \numprint{9.23}       & \numprint{0.00} & \numprint{0.08}\\
-Multitry & \numprint{1.71}  & \numprint{1.49}  & \numprint{5.55}       & \numprint{1.21} & \numprint{1.30}\\
-Cyc      & \numprint{2.42}  & \numprint{1.95}  & \numprint{3.27}       & \numprint{1.25} & \numprint{1.41}\\
-MB       & \numprint{3.35}  & \numprint{2.64}  & \numprint{2.92}       & \numprint{1.82} & \numprint{1.91}\\
-Flow     & \numprint{9.36}  & \numprint{7.87}  & \numprint{1.66}       & \numprint{6.18} & \numprint{6.08}\\
\hline
\end{tabular}

\vspace*{0.3cm}

\begin{tabular}{|r|rrr||rr| } 
\hline
Variant & Avg. & Best. & $t$ & Eff. Avg. & Eff. Best.  \\
\hline
Strong    &   \numprint{2683}  & \numprint{2617}  & \numprint{8.93}    & \numprint{2636} & \numprint{2616} \\
                  \hline
-KWay     &   \numprint{-0,04} & \numprint{-0,11} & \numprint{9,23}    & \numprint{0.00} & \numprint{0.08} \\
-Multitry &   \numprint{1,27}  & \numprint{1,11}  & \numprint{5,52}    & \numprint{0.83} & \numprint{0.99} \\
-MB       &   \numprint{0,26}  & \numprint{0,08}  & \numprint{8,34}    & \numprint{0.11} & \numprint{0.11} \\
-Flow     &   \numprint{1,53}  & \numprint{0,99}  & \numprint{6,33}    & \numprint{0.87} & \numprint{0.80} \\
\hline
\end{tabular}

\end{center}
\caption{Removal tests (top): each configuration is same as its predecessor minus the component shown at beginning of the row. Knockout tests (bottom): each configuration is same as KaFFPa Strong minus the component shown at beginning of the row. All average cuts and best cuts are shown as increases in cut (\%) relative to the values obtained by KaFFPa Strong.}
\label{tab:removaltestresultscompressed}
\end{table}
\vspace*{-1.5cm}
\subsection{Comparison with other Partitioners}\label{ss:others}
We now switch to our suite of larger graphs since that's what \algname\ was
designed for and because we thus avoid the effect of overtuning our algorithm
parameters to the instances used for calibration. 
We compare ourselves with KaSPar Strong, KaPPa Strong, DiBaP Strong, Scotch and Metis. 

Figure~\ref{fig:comparisonothergrafic} summarizes the results.
We excluded the European and German road network 
as well as the Random Geometric Graph for the comparison with DiBaP since DiBaP can't handle singletons. In general, we excluded the case $k=2$ for the European road network for the comparison since it runs out of memory for this case. 
As recommended by Henning Meyerhenke DiBaP was run with 3 bubble repetitions, 10 FOS/L consolidations and 14 FOS/L iterations.
Detailed per instance results can be found in Appendix Table~\ref{tab:detailedperinstancebasis}. 

kMetis produces about 33\% larger cuts than the strong variant of KaFFPa. 
Scotch, DiBaP, KaPPa, and KaSPar produce 20\%,11\%, 12\% and 3\% larger cuts than KaFFPa respectively. 
The strong variant of KaFFPa now produces the average best cut results of KaSPar on average (which where obtained using five repeated executions of KaSPar). In 57 out of 66 cases KaFFPa produces a better best cut than the best cut obtained by KaSPar.  

The largest absolute improvement to KaSPar Strong is obtained on \textit{af\_shell10} at $k=16$ where the best cut produced by KaSPar-Strong is 7.2\% larger than the best cut produced by KaFFPa Strong. 
The largest absolute improvement to kMetis is obtained on the European road network where kMetis produces cuts that are a factor 5.5 larger than the edge cuts produces by our strong configuration. 
\begin{figure}[t]
\begin{center}
\begin{minipage}{6cm}
\includegraphics[width=6cm]{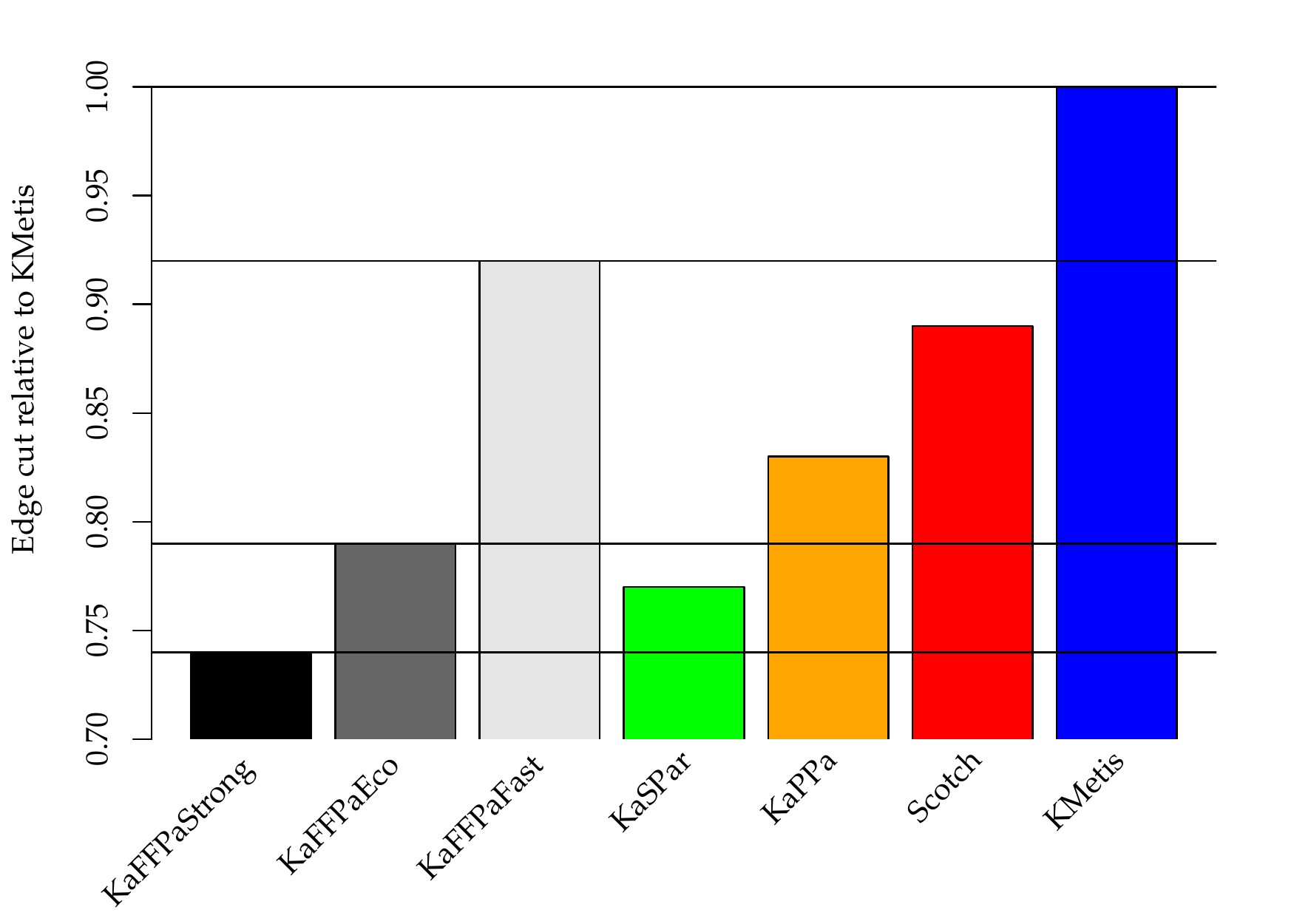} 
\end{minipage}
\begin{minipage}{6cm}
          \begin{tabular}{|l||r|r|r||}
\hline
algorithm & \multicolumn{3}{|c||}{large graphs}\\
             \hline
              &  best   &  avg.   & t[s]  \\\hline
KaFFPa Strong &  12\,054& 12\,182 & 121.50 \\
KaSPar Strong &  12\,450& +3\%    & 87.12 \\
KaFFPa Eco    &  12\,763& +6\%    & 3.82 \\
KaPPa Strong  &  13\,323& +12\%   & 28.16 \\
Scotch        &  14\,218& +20\%   &  3.55 \\ 
KaFFPa Fast   &  15\,124& +24\%   &  0.98 \\
kMetis        &  15\,167& +33\%   &  0.83 \\
              \hline
\end{tabular}
\end{minipage}

\end{center}
        \caption{Averaged quality of the different partitioning algorithms. }
        \label{fig:comparisonothergrafic}
\end{figure}

The eco configuration of KaFFPa now outperforms Scotch and DiBaP being than DiBaP while producing 4.7 \% and 12\% smaller cuts than DiBap and Scotch respectively.
The run time difference to both algorithms gets larger with increasing number of blocks. 
Note that DiBaP has a factor 3 larger run times than KaFFPa Eco on average and up to factor 4 on average for $k=64$.

On the largest graphs available to us (\textit{delaunay, rgg, eur}) KaFFPa Fast outperforms KMetis in terms of quality and runtime. 
For example on the \textit{european road network} kMetis has about 44\% larger run times and produces up to a factor 3 (for $k=16$) larger cuts. 

We now turn into graph sequence tests. Here we take two graph families (\textit{rgg}, \textit{delaunay}) and study the behaviour of our algorithms when the graph size increases. 
In Figure~\ref{fig:regressionrgg}, we see for increasing size of random geometric graphs the run time advantage of KaFFPa Fast relative to kMetis increases.  
The largest difference is obtained on the largest graph where kMetis has 70\% larger run times than our fast configuration which still produces 2.5\% smaller cuts.
We observe the same behaviour for the delaunay based graphs (see appendix for more details). Here we get a run time advantage of up to 24\% with 6.5\% smaller cuts for the largest graph.
Also note that for these graphs the improvement of KaFFPa Strong and Eco in terms of quality relative to kMetis increases with increasing graph size (up to 32\% for delaunay and up to 47\% for rgg for our strong configuration). 

\subsection{The Walshaw Benchmark}\label{ss:benchmark}
We now apply \algname\ to Walshaw's benchmark archive \cite{walshaw2000mpm} using the rules used there, i.e., running time is no
issue but we want to achieve minimal cut values for $k\in \set{2, 4, 8, 16, 32, 64}$ and balance parameters $\epsilon\in\set{0,0.01,0.03,0.05}$.
We tried all combinations except the case $\epsilon=0$ because flows are not made for this case. 

We ran KaFFPa Strong with a time limit of two hours per graph and $k$ and report the best result obtained in the appendix. KaFFPa computed 317 partitions which are better that previous best partitions reported
there: 99 for 1\%, 108 for 3\% and 110 for 5\%. Moreover, it reproduced equally sized cuts in 118 of the 295 remaining cases.
The complete list of improvements is available at Walshaw's archive \cite{walshaw2000mpm}.
We obtain only a few improvements for $k=2$. However, in this case we are able to reproduce the currently best result in 91 out of 102 cases. For the large graphs (using 78000 nodes as a cut off) we obtain cuts that are lower or equal to the current entry in 92\% of the cases. 
The biggest absolute improvement is observed for instance \textit{add32} (for each imbalance) and $k=4$ where the old partitions cut 10 \% more edges. 
The biggest absolute difference is obtained for \textit{m14b} at 3 \% imbalance and $k=64$ where the new partition cuts 3183 less edges. 

After the partitions were accepted, we ran KaFFPa Strong as before and took the previous entry as input. Now in 560 out of 612 cases we where able to improve a given entry or have been able to reproduce the current result.
\begin{figure}[t!]
\centering
\includegraphics[width=6cm]{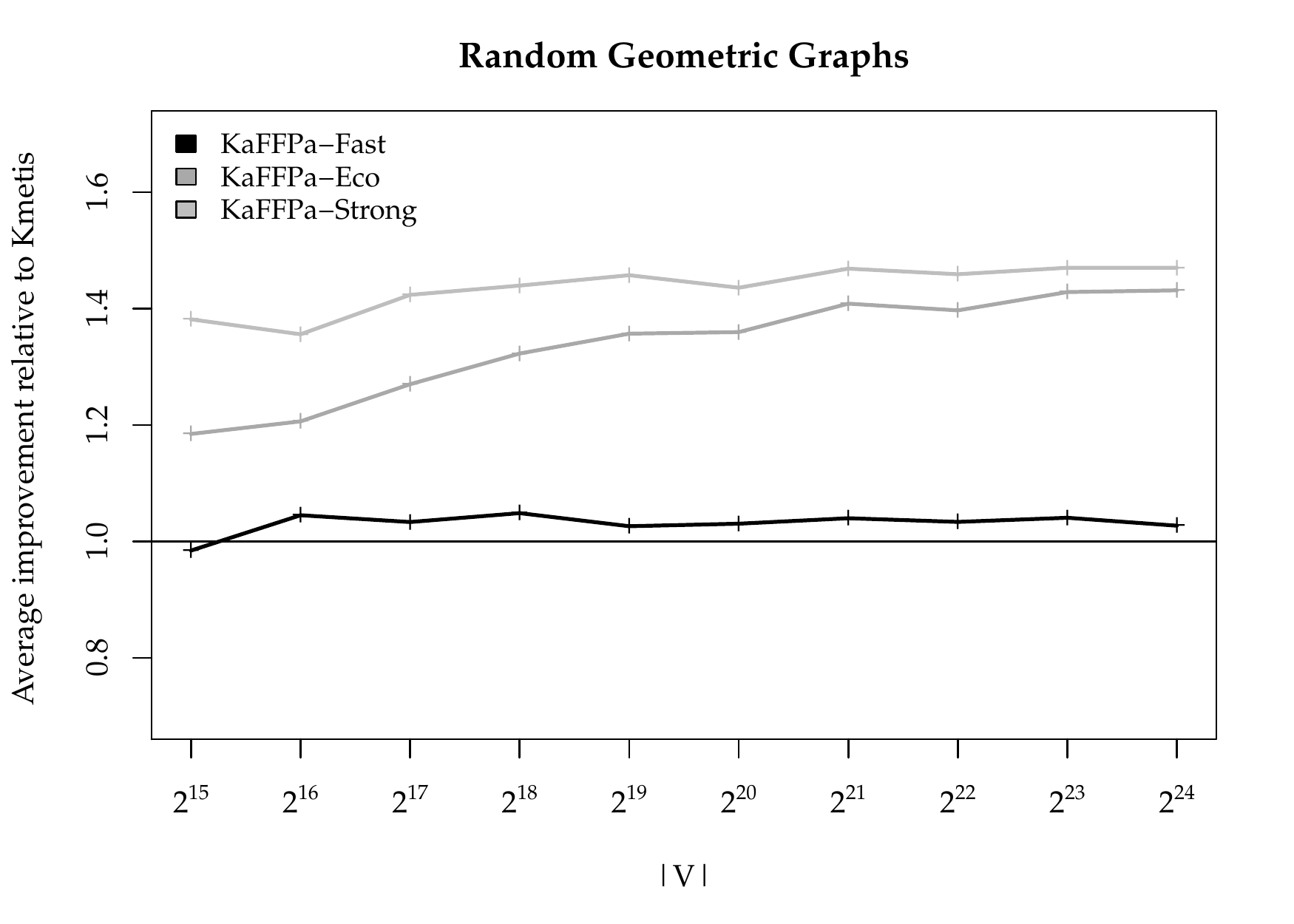}  \includegraphics[width=6cm]{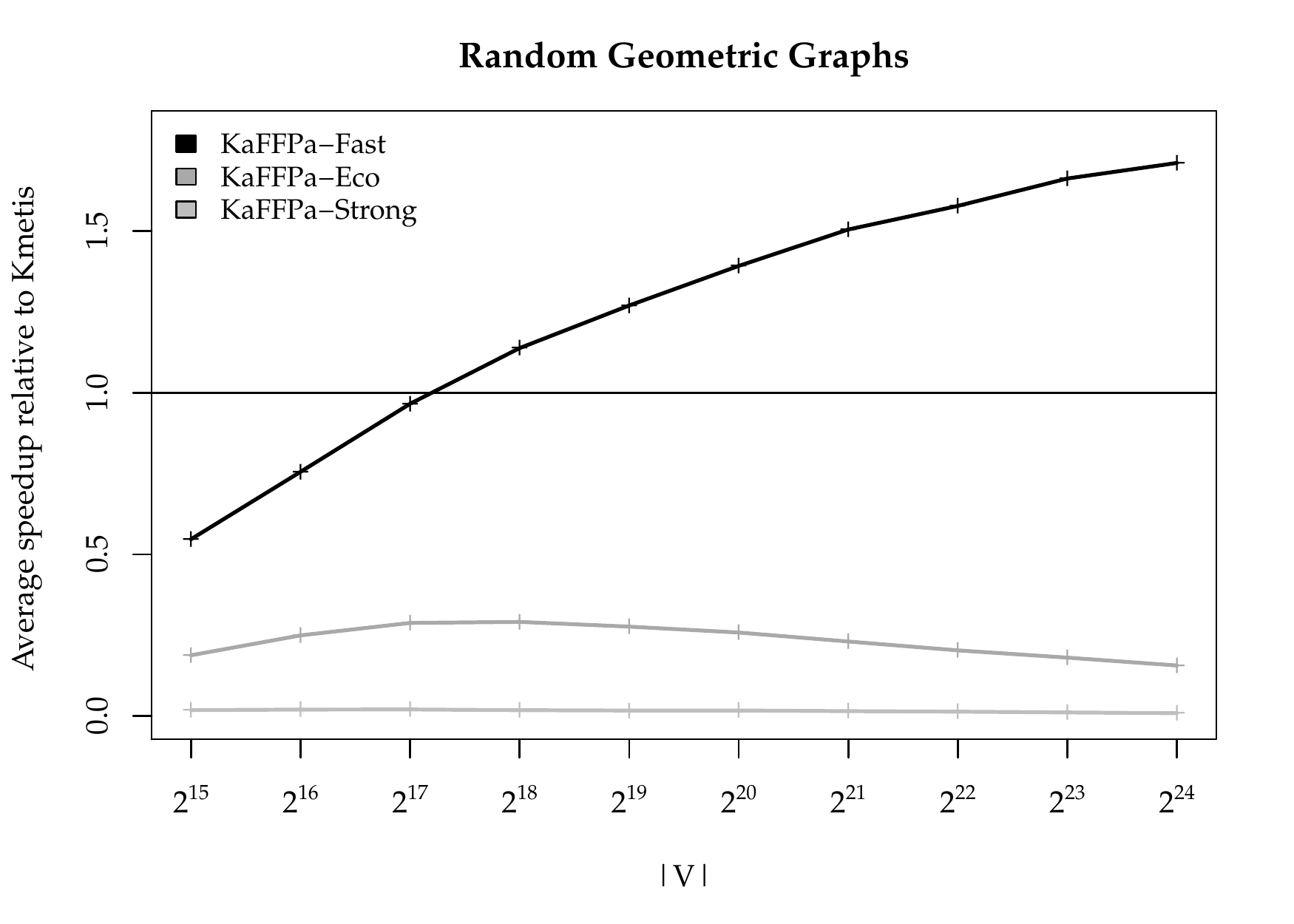}
\caption{Graph sequence test for Random Geometric Graphs.}
\label{fig:regressionrgg}
\end{figure}

\section{Conclusions and Future Work}\label{s:conclusions}
KaFFPa is an approach to graph partitioning which currently computes the best known partitions for many graphs, at least when a certain imbalance is allowed. 
This success is due to new local improvement methods, which are based on max-flow min-cut computations and more localized local searches, and global search strategies which were transferred from multigrid linear solvers.  

A lot of opportunities remain to further improve KaFFPa. 
For example we did not try to handle the case $\epsilon=0$ since this may require different local search strategies. 
Furthermore, we want to try other initial partitioning algorithms and ways to integrate KaFFPa into other metaheuristics like evolutionary search. 

Moreover, we would like to go back to parallel graph partitioning. Note that our max-flow min-cut local improvement methods fit very well into the parallelization scheme of KaPPa \cite{kappa}. 
We also want to combine KaFFPa with the $n$-level idea from KaSPar \cite{kaspar}.
Other refinement algorithms, e.g., based on diffusion or MQI could be
tried within our framework of pairwise refinement.

The current implementation of \algname\ is a research prototype rather than
a widely usable tool. However, we are planing an open source release available for download.  \\

\subsection*{Acknowledgements} We would like to thank Vitaly Osipov for supplying data for KaSPar and Henning Meyerhenke for 
providing a DiBaP-full executable. 
We also thank Tanja Hartmann, Robert G\"orke and Bastian Katz for valuable advice regarding balanced min cuts.

{
\bibliographystyle{plain}
\bibliography{diss,quellen}}
\appendix
\vfill
\pagebreak
\begin{figure*}[h!]
\small
\begin{center}
        \begin{algorithmic}
        \STATE \textbf{procedure} \textit{W-Cycle}(G)
        \STATE \quad $G' = $coarsen$(G)$
        \STATE \quad \textbf{if} $G'$ small enough \textbf{then}
        \STATE \quad \quad initial partition $G'$ if not partitioned
        \STATE \quad \quad apply partition of $G'$ to $G$
        \STATE \quad \quad perform refinement on $G$
        \STATE \quad \textbf{else}
        \STATE \quad \quad W-Cycle($G'$) and apply partition to $G$

        \STATE  \quad \quad  perform refinement on $G$
        \STATE  \quad \quad $G'' = $coarsen$(G)$
        \STATE  \quad \quad W-Cycle($G''$) and apply partition to $G$
        \STATE  \quad \quad perform refinement on $G$
\end{algorithmic}
        \begin{algorithmic}
        \STATE \textbf{procedure} \textit{F-Cycle}(G)
        \STATE \quad $G' = $coarsen$(G)$
        \STATE \quad \textbf{if} $G'$ small enough \textbf{then}
        \STATE \quad \quad initial partition $G'$ if not partitioned
        \STATE \quad \quad apply partition of $G'$ to $G$
        \STATE \quad \quad perform refinement on $G$
        \STATE \quad \textbf{else}
        \STATE \quad \quad F-Cycle($G'$) and apply partition to $G$
        \STATE  \quad \quad  perform refinement on $G$
        \STATE  \quad \quad \textbf{if} no. trails. calls on cur. level < 2 \textbf{then}
        \STATE  \quad \quad \quad  $G'' = $coarsen$(G)$
        \STATE  \quad \quad \quad F-Cycle($G''$) and apply partition to $G$
        \STATE  \quad \quad \quad perform refinement on $G$
        \end{algorithmic}
\end{center}
\caption{Pseudocode for the different global search strategies.}
\label{fig:globalsearchpseudocode}
\end{figure*}
\begin{figure}[h!]
\begin{algorithmic}
\STATE \textbf{procedure} \textit{activeBlockScheduling}()
\STATE   \quad set all blocks active
\STATE   \quad \textbf{while} there are active blocks
\STATE   \quad \quad A := <edge (u,v) in quotient graph : u active or v active>
\STATE   \quad \quad set all blocks inactive
\STATE   \quad \quad \textbf{permute} A randomly
\STATE   \quad \quad  \textbf{for each} (u,v) in A \textbf{do}
\STATE   \quad \quad  \quad   pairWiseImprovement(u,v) 
\STATE   \quad \quad  \quad   multitry FM search starting with boundary of u and v
\STATE   \quad \quad  \quad \textbf{if} anything changed during local search \textbf{then}
\STATE   \quad \quad  \quad  \quad  activate blocks that have changed during pairwise 
\STATE   \quad \quad  \quad  \quad  or multitry FM search
\end{algorithmic}
\caption{Pseudocode for the active block scheduling algorithm. In our implementation the pairwise improvement step starts with a FM local search which is followed by a max-flow min-cut based improvement.}
\label{fig:activeblockscheduling}
\end{figure}

\newpage
 \begin{table}
\small
\begin{center}
\begin{tabular}{|l|r|r|}
\hline
\multicolumn{3}{|l|}{Medium sized instances} \\
\hline
graph & $n$ & $m$ \\
\hline
rgg17	 & $2^{17}$ 	& 1\,457\,506\\
rgg18 	 & $2^{18}$  	& 3\,094\,566\\
Delaunay17	 & $2^{17}$	& 786\,352\\
Delaunay18 	 & $2^{18}$ 	& 1\,572\,792\\
\hline
bcsstk29 	&	13\,992 &	605\,496\\
4elt 		&	15\,606 & 	91\,756\\
fesphere  	& 	16\,386 & 	98\,304\\
cti  		&	16\,840 & 	96\,464\\	
memplus 	&	17\,758 &	108\,384\\
cs4 	 	&       33\,499 & 	87\,716\\
pwt 		&	36\,519 & 	289\,588\\
bcsstk32 	&	44\,609 &	1\,970\,092\\	
body 		& 45\,087 	& 327\,468\\
t60k 		& 60\,005 	& 178\,880\\
wing  		& 62\,032 	& 243\,088\\
finan512 	& 74\,752 	& 522\,240\\	
rotor & 99\,617 & 1\,324\,862 \\
\hline
bel & 463\,514 & 1\,183\,764\\
nld & 893\,041 & 2\,279\,080\\
\hline
af\_shell9 	& 504\,855 & 17\,084\,020\\
\hline
\hline
\multicolumn{3}{|l|}{Large instances} \\
\hline
rgg20 		& $2^{20}$ 	& 13\,783\,240   \\
Delaunay20 		& $2^{20}$ 	& 12\,582\,744  \\
\hline
fetooth	 	& 78\,136  	& 905\,182\\
598a  		& 110\,971 	& 1\,483\,868\\
ocean 		& 143\,437 	& 819\,186\\
144  		& 144\,649 	& 2\,148\,786\\
wave 		& 156\,317 	& 2\,118\,662  \\
m14b 		& 214\,765 	& 3\,358\,036 \\
auto 		& 448\,695 	& 6\,629\,222 \\
\hline
deu & 4\,378\,446 & 10\,967\,174  \\
eur & 18\,029\,721 & 44\,435\,372 \\
\hline
af\_shell10 & 1\,508\,065 & 51\,164\,260  \\
\hline
\end{tabular}
\end{center}
\caption{
 Basic properties of the graphs from our 
 benchmark set. The large instances are split into four groups:
 geometric graphs, FEM graphs, street networks, sparse matrices.  Within their groups, the graphs are sorted by size.  }
\label{tab:instances}
\end{table}
\begin{table*}[h!]
\begin{center}
\hspace*{-1.25cm} 
\small
\begin{tabular}{|l|r|r|r|r|r|r|r|r|r|r|r|r|r|r|r|r|}\hline
			Variant & \multicolumn{4}{c|}{(+Flow, -MB, -FM )} & \multicolumn{4}{c|}{(+Flow, +MB, -FM)} & \multicolumn{4}{c|}{(+Flow, -MB, +FM)} &  \multicolumn{4}{c|}{(+Flow, +MB, +FM)}\\
                        \hline
			$\alpha' $        & Avg.  & Best.  & Bal. & $t$  & Avg.  & Best.  &  Bal. & $t$  & Avg.  & Best.  & Bal. & $t$ & Avg.  & Best.  & Bal. & $t$ \\
			\hline 
			$16$ & \numprint{3031} & \numprint{2888} & \numprint{1.025} & \numprint{4.17} & \numprint{2950} & \numprint{2841} & \numprint{1.023} & \numprint{3.92} & \numprint{2802} & \numprint{2704} & \numprint{1.025} & \numprint{4.30} & \numprint{2774} & \numprint{2688} & \numprint{1.023} & \numprint{5.01}\\
			$8$  & \numprint{3044} & \numprint{2905} & \numprint{1.025} & \numprint{2.11} & \numprint{2962} & \numprint{2855} & \numprint{1.023} & \numprint{2.07} & \numprint{2806} & \numprint{2705} & \numprint{1.025} & \numprint{2.41} & \numprint{2778} & \numprint{2693} & \numprint{1.023} & \numprint{2.72}\\
			$4$  & \numprint{3126} & \numprint{2963} & \numprint{1.024} & \numprint{1.24} & \numprint{3041} & \numprint{2933} & \numprint{1.021} & \numprint{1.29} & \numprint{2825} & \numprint{2723} & \numprint{1.025} & \numprint{1.62} & \numprint{2800} & \numprint{2706} & \numprint{1.022} & \numprint{1.76}\\
			$2$  & \numprint{3374} & \numprint{3180}& \numprint{1.022} & \numprint{0.90} & \numprint{3274} & \numprint{3107} & \numprint{1.018} & \numprint{0.96} & \numprint{2869} & \numprint{2758} & \numprint{1.024} & \numprint{1.31} & \numprint{2855} & \numprint{2746} & \numprint{1.021} & \numprint{1.39}\\
			$1$  & \numprint{3698} & \numprint{3488} & \numprint{1.018} & \numprint{0.76} & \numprint{3587} & \numprint{3410} & \numprint{1.016} & \numprint{0.80} & \numprint{2926} & \numprint{2804} & \numprint{1.024} & \numprint{1.19} & \numprint{2923} & \numprint{2802} & \numprint{1.023} & \numprint{1.22}\\
\hline
\end{tabular}

\end{center}
\caption{The final score of different algorithm configurations compared against the basic two-way FM configuration. Here $\alpha'$ is the flow region upper bound factor. The values are average values as described in Section~\ref{s:experiments}.}
\label{tab:flowcompabs}
\end{table*}

\begin{table*}
\begin{center}
	
\small
\begin{tabular}{|l|r|r|r|}\hline
			Effectiveness  & \multicolumn{3}{c|}{}\\
                        (+Flow, +MB, -FM) & Avg. & Best. & Bal.\\
			\hline 
			$\alpha'=1$&\numprint{3389}&\numprint{3351}&\numprint{1.016}  \\
			2          &\numprint{3088}&\numprint{3049}&\numprint{1.017} \\
			4          &\numprint{2922}&\numprint{2892}&\numprint{1.022} \\
			8          &\numprint{2865}&\numprint{2841}&\numprint{1.023} \\
			16         &\numprint{2870}&\numprint{2839}&\numprint{1.023} \\
\hline
                        (-Flow, -MB, +FM)&\numprint{2833}&\numprint{2803}&\numprint{1.025}  \\
			\hline

\end{tabular} \\
\vspace*{0.3cm}
\begin{tabular}{|l|r|r|r|r|}\hline
          Effectiveness & \multicolumn{3}{c|}{}\\  
          (+Flow,-MB, +FM) & Avg. & Best. & Bal.\\

         \hline
                 $\alpha'=1$& \numprint{2786} & \numprint{2759} & \numprint{1.024}\\
                 2          & \numprint{2748} & \numprint{2724} & \numprint{1.024}\\
                 4          & \numprint{2721} & \numprint{2698} & \numprint{1.025}\\
                 8          & \numprint{2718} & \numprint{2690} & \numprint{1.025}\\
                 16         & \numprint{2730} & \numprint{2697} & \numprint{1.025}\\
\hline
                 (-Flow, -MB, +FM)& \numprint{2831} & \numprint{2801} & \numprint{1.025}\\
                 \hline

                 \end{tabular} \\

\vspace*{0.3cm}
		\begin{tabular}{|l|r|r|r|}\hline
                        Effectiveness & \multicolumn{3}{c|}{} \\
			(+Flow,+MB,+FM) & Avg. & Best. & Bal.\\
			\hline 
			$\alpha'=1$      & \numprint{2781} & \numprint{2754} & \numprint{1.023}\\
			2                & \numprint{2735} & \numprint{2711} & \numprint{1.021}\\
			4                & \numprint{2702} & \numprint{2682} & \numprint{1.022}\\
			8                & \numprint{2699} & \numprint{2675} & \numprint{1.023}\\
			16               & \numprint{2711} & \numprint{2682} & \numprint{1.022}\\
\hline
		        (-Flow, -MB, +FM)& \numprint{2827} & \numprint{2799} & \numprint{1.025}\\
			\hline

		\end{tabular}

\end{center}
\caption{Each table is the result of an effectiveness test for six different algorithm configurations. All values are average values as described in Section~\ref{s:experiments}. }
\label{tab:flowcompeff}
\end{table*}

\begin{table*}
\begin{center}
	
\small
\begin{tabular}{|l|r|r|r|r||r|r|}\hline

                  Algorithm & Avg. & Best. & Bal. &$t$ & Eff. Avg. & Eff. Best. \\

                  \hline
                  2 F-cycle        & \numprint{2895} & \numprint{2773} & \numprint{1.023} & \numprint{2.31} & \numprint{2806} & \numprint{2760}\\
                  3 V-cycle        & \numprint{2895} & \numprint{2776} & \numprint{1.023} & \numprint{2.49} & \numprint{2810} & \numprint{2766}\\
                  2 W-cycle        & \numprint{2889} & \numprint{2765} & \numprint{1.024} & \numprint{2.77} & \numprint{2810} & \numprint{2760}\\
                  1 W-cycle        & \numprint{2934} & \numprint{2810} & \numprint{1.024} & \numprint{1.38} & \numprint{2815} & \numprint{2773}\\
                  1 F-cycle        & \numprint{2941} & \numprint{2813} & \numprint{1.024} & \numprint{1.18} & \numprint{2816} & \numprint{2783}\\
                  2 V-cycle        & \numprint{2918} & \numprint{2796} & \numprint{1.024} & \numprint{1.67} & \numprint{2817} & \numprint{2778}\\
                  1 V-cycle        & \numprint{2973} & \numprint{2841} & \numprint{1.024} & \numprint{0.85} & \numprint{2834} & \numprint{2801}\\
                  \hline
\end{tabular}
\end{center}
\caption{Test results for normal and effectiveness tests for different global search strategies and different parameters. }
\label{tab:globalsearchresultsabs}
\end{table*}

\begin{table*}[h!]         
\begin{center}
\small
\hspace*{-1.3cm}
\begin{tabular}{|r|rrr|rrr|rrr|rrr|rrr|rrr|rrr|rrr|rrr|rrr| } 
\hline
$k$ & \multicolumn{3}{c}{Strong} & \multicolumn{3}{c}{-Kway}  & \multicolumn{3}{c}{-Multitry} & \multicolumn{3}{c}{-Cyc}& \multicolumn{3}{c}{-MB} & \multicolumn{3}{c|}{-Flow}  \\
\hline
 &Avg. & Best. & $t$   & Avg. & Best. & $t$  & Avg. & Best. & $t$ & Avg. & Best. & $t$& Avg. & Best. & $t$&  Avg. & Best. & $t$   \\
\hline
2   & \numprint{561}  & \numprint{548}  & \numprint{2.85}  & \numprint{561}  & \numprint{548}  & \numprint{2.87}  &  \numprint{564} & \numprint{549}  & \numprint{2.68}  &   \numprint{568}  & \numprint{549}  & \numprint{1.42}  &  \numprint{575} & \numprint{551}  & \numprint{1.33}  &  \numprint{627} & \numprint{582}  & \numprint{0.85}  \\  
4   & \numprint{1286} & \numprint{1242} & \numprint{5.13}  & \numprint{1287} & \numprint{1236} & \numprint{5.28}  & \numprint{1299} & \numprint{1244} & \numprint{4.26}  &   \numprint{1305} & \numprint{1248} & \numprint{2.40}  & \numprint{1317} & \numprint{1254} & \numprint{2.18}  & \numprint{1413} & \numprint{1342} & \numprint{1.02} \\ 
8   & \numprint{2314} & \numprint{2244} & \numprint{7.52}  & \numprint{2314} & \numprint{2241} & \numprint{7.82}  & \numprint{2345} & \numprint{2273} & \numprint{5.34}  &   \numprint{2356} & \numprint{2279} & \numprint{3.11}  & \numprint{2375} & \numprint{2295} & \numprint{2.70}  & \numprint{2533} & \numprint{2441} & \numprint{1.32} \\ 
16  & \numprint{3833} & \numprint{3746} & \numprint{11.26} & \numprint{3829} & \numprint{3735} & \numprint{11.73} & \numprint{3907} & \numprint{3813} & \numprint{6.40}  &   \numprint{3937} & \numprint{3829} & \numprint{3.79}  & \numprint{3970} & \numprint{3867} & \numprint{3.32}  & \numprint{4180} & \numprint{4051} & \numprint{1.80} \\
\hline
\end{tabular}

\vspace*{0.3cm}

\begin{tabular}{|r|rr|rr|rr|rr|rr|rr|rr|rr|rr|rr|rr|rr|rr|rr|rr| } 
\hline
Effectiveness & \multicolumn{2}{c}{Strong} & \multicolumn{2}{c}{-Kway}  & \multicolumn{2}{c}{-Multitry} & \multicolumn{2}{c}{-Cyc}& \multicolumn{2}{c}{-MB} & \multicolumn{2}{c|}{-Flow}  \\
\hline

 $k$&Avg. & Best.  & Avg. & Best. & Avg. & Best.& Avg. & Best.& Avg. & Best.&  Avg. & Best.\\
\hline

2     & \numprint{550}  & \numprint{547}  & \numprint{550}  & \numprint{548}  & \numprint{550}  & \numprint{548}  &  \numprint{549}  & \numprint{548}  & \numprint{552}  & \numprint{549}  & \numprint{581}  & \numprint{573}   \\ 
4     & \numprint{1251} & \numprint{1240} & \numprint{1251} & \numprint{1243} & \numprint{1257} & \numprint{1246} &  \numprint{1255} & \numprint{1245} & \numprint{1263} & \numprint{1252} & \numprint{1316} & \numprint{1299}   \\ 
8     & \numprint{2263} & \numprint{2242} & \numprint{2270} & \numprint{2249} & \numprint{2280} & \numprint{2267} &  \numprint{2277} & \numprint{2263} & \numprint{2289} & \numprint{2273} & \numprint{2408} & \numprint{2387}   \\ 
16    & \numprint{3773} & \numprint{3745} & \numprint{3769} & \numprint{3742} & \numprint{3830} & \numprint{3795} &  \numprint{3828} & \numprint{3799} & \numprint{3846} & \numprint{3813} & \numprint{4029} & \numprint{3996}   \\ 
32    & \numprint{6000} & \numprint{5943} & \numprint{6001} & \numprint{5947} & \numprint{6116} & \numprint{6078} &  \numprint{6139} & \numprint{6099} & \numprint{6170} & \numprint{6128} & \numprint{6403} & \numprint{6369}   \\ 
64    & \numprint{9523} & \numprint{9463} & \numprint{9502} & \numprint{9437} & \numprint{9745} & \numprint{9702} &  \numprint{9811} & \numprint{9754} & \numprint{9881} & \numprint{9829} & \numprint{10139}& \numprint{10085}  \\
\hline 
\end{tabular}

\end{center}
\caption{Removal tests: each configuration is same as left neighbor minus the component shown at the top of the column. The first table shows detailed results for all $k$ in a normal test. The second table shows the results for an effectivity test. }
\end{table*}

\begin{table*}[h!]
\begin{center}
\small

\hspace*{-1.3cm}
\begin{tabular}{|r|rrr|rrr|rrr|rrr|rrr|rrr|rrr|rrr|rrr|rrr| } 
\hline
$k$ & \multicolumn{3}{c}{Strong} & \multicolumn{3}{c}{-Kway}  & \multicolumn{3}{c}{-Multitry} & \multicolumn{3}{c}{-Cyc}& \multicolumn{3}{c}{-MB} & \multicolumn{3}{c|}{-Flow}  \\
\hline
 &Avg. & Best. & $t$   & Avg. & Best. & $t$  & Avg. & Best. & $t$ & Avg. & Best. & $t$& Avg. & Best. & $t$&  Avg. & Best. & $t$   \\
\hline
2   & \numprint{561}  & \numprint{548}  & \numprint{2.85}  & \numprint{0.00}  & \numprint{0.00}  & \numprint{2.87}  & \numprint{0.53} & \numprint{0.18} & \numprint{2.68} & \numprint{1.25} & \numprint{0.18} & \numprint{1.42} & \numprint{2.50} & \numprint{0.55} & \numprint{1.33} & \numprint{11.76} & \numprint{6.20} & \numprint{0.85} \\
4   & \numprint{1286} & \numprint{1242} & \numprint{5.13}  & \numprint{0.08}  & \numprint{-0.48} & \numprint{5.28}  & \numprint{1.01} & \numprint{0.16} & \numprint{4.26} & \numprint{1.48} & \numprint{0.48} & \numprint{2.40} & \numprint{2.41} & \numprint{0.97} & \numprint{2.18} & \numprint{9.88}  & \numprint{8.05} & \numprint{1.02} \\
8   & \numprint{2314} & \numprint{2244} & \numprint{7.52}  & \numprint{0.00}  & \numprint{-0.13} & \numprint{7.82}  & \numprint{1.34} & \numprint{1.29} & \numprint{5.34} & \numprint{1.82} & \numprint{1.56} & \numprint{3.11} & \numprint{2.64} & \numprint{2.27} & \numprint{2.70} & \numprint{9.46}  & \numprint{8.78} & \numprint{1.32} \\
16  & \numprint{3833} & \numprint{3746} & \numprint{11.26} & \numprint{-0.10} & \numprint{-0.29} & \numprint{11.73} & \numprint{1.93} & \numprint{1.79} & \numprint{6.40} & \numprint{2.71} & \numprint{2.22} & \numprint{3.79} & \numprint{3.57} & \numprint{3.23} & \numprint{3.32} & \numprint{9.05}  & \numprint{8.14} & \numprint{1.80} \\
\hline
\end{tabular}

\vspace*{0.3cm}


\begin{tabular}{|r|rr|rr|rr|rr|rr|rr|rr|rr|rr|rr|rr|rr|rr|rr|rr| } 
\hline
Effectiveness & \multicolumn{2}{c}{Strong} & \multicolumn{2}{c}{-Kway}  & \multicolumn{2}{c}{-Multitry} & \multicolumn{2}{c}{-Cyc}& \multicolumn{2}{c}{-MB} & \multicolumn{2}{c|}{-Flow}  \\
\hline

 $k$&Avg. & Best.  & Avg. & Best. & Avg. & Best.& Avg. & Best.& Avg. & Best.&  Avg. & Best.\\
\hline

2     & \numprint{550}  & \numprint{547}  & \numprint{0.00}  & \numprint{0.18}  & \numprint{0.00} & \numprint{0.18} & \numprint{-0.18} & \numprint{0.18} & \numprint{0.36} & \numprint{0.37} & \numprint{5.64} & \numprint{4.75} \\
4     & \numprint{1251} & \numprint{1240} & \numprint{0.00}  & \numprint{0.24}  & \numprint{0.48} & \numprint{0.48} & \numprint{0.32}  & \numprint{0.40} & \numprint{0.96} & \numprint{0.97} & \numprint{5.20} & \numprint{4.76} \\
8     & \numprint{2263} & \numprint{2242} & \numprint{0.31}  & \numprint{0.31}  & \numprint{0.75} & \numprint{1.12} & \numprint{0.62}  & \numprint{0.94} & \numprint{1.15} & \numprint{1.38} & \numprint{6.41} & \numprint{6.47} \\
16    & \numprint{3773} & \numprint{3745} & \numprint{-0.11} & \numprint{-0.08} & \numprint{1.51} & \numprint{1.34} & \numprint{1.46}  & \numprint{1.44} & \numprint{1.93} & \numprint{1.82} & \numprint{6.79} & \numprint{6.70} \\
32    & \numprint{6000} & \numprint{5943} & \numprint{0.02}  & \numprint{0.07}  & \numprint{1.93} & \numprint{2.27} & \numprint{2.32}  & \numprint{2.62} & \numprint{2.83} & \numprint{3.11} & \numprint{6.72} & \numprint{7.17} \\
64    & \numprint{9523} & \numprint{9463} & \numprint{-0.22} & \numprint{-0.27} & \numprint{2.33} & \numprint{2.53} & \numprint{3.02}  & \numprint{3.08} & \numprint{3.76} & \numprint{3.87} & \numprint{6.47} & \numprint{6.57} \\
\hline 
Avg.  & \numprint{2636} & \numprint{2616} & \numprint{0.00}  & \numprint{0.08}  & \numprint{1.21} & \numprint{1.30} & \numprint{1.25}  & \numprint{1.41} & \numprint{1.82} & \numprint{1.91} & \numprint{6.18} & \numprint{6.08} \\
\hline 
\end{tabular}

\end{center}
\caption{Removal tests: each configuration is same as its left neighbor minus the component shown at the top of the column. The first table shows detailed results for all $k$ in a normal test. The second table shows the results for an effectivity test. All values are increases in cut are relative to the values obtained by KaFFPa Strong. }
\label{tab:removaltestresults}
\end{table*}

\begin{table*}[h!] 
\begin{center}
        
\small

\begin{tabular}{|r|rrr|rrr|rrr|rrr|rrr|rrr| } 
\hline
$k$ & \multicolumn{3}{c}{Strong} & \multicolumn{3}{c}{-Kway}  & \multicolumn{3}{c}{-Multitry} & \multicolumn{3}{c}{-MB}& \multicolumn{3}{c|}{-Flows}  \\
\hline

 & Avg. & Best. & $t$ & Avg. & Best. & $t$ & Avg. & Best. & $t$ & Avg. & Best. & $t$& Avg. & Best. & $t$  \\
\hline 
2 &
\numprint{561} & \numprint{548} & \numprint{2.85}    &
\numprint{561} & \numprint{548} & \numprint{2.86}    &
\numprint{561} & \numprint{548} & \numprint{2.72}    &
\numprint{564} & \numprint{548} & \numprint{2.70}    &
\numprint{582} & \numprint{559} & \numprint{1.94}   
\\
4 &
\numprint{1286} & \numprint{1242} & \numprint{5.14}    &
\numprint{1287} & \numprint{1236} & \numprint{5.29}    &
\numprint{1293} & \numprint{1240} & \numprint{4.23}    &
\numprint{1290} & \numprint{1239} & \numprint{4.68}    &
\numprint{1312} & \numprint{1252} & \numprint{2.95}    
\\
8 &
\numprint{2314} & \numprint{2244} & \numprint{7.52}    &
\numprint{2314} & \numprint{2241} & \numprint{7.81}    &
\numprint{2337} & \numprint{2271} & \numprint{5.24}    &
\numprint{2322} & \numprint{2249} & \numprint{6.88}    &
\numprint{2347} & \numprint{2270} & \numprint{4.88}    
\\
16 &
\numprint{3833} & \numprint{3746} & \numprint{11.19}    &
\numprint{3829} & \numprint{3735} & \numprint{11.69}    &
\numprint{3894} & \numprint{3799} & \numprint{6.27}    &
\numprint{3838} & \numprint{3747} & \numprint{10.41}    &
\numprint{3870} & \numprint{3779} & \numprint{8.22}    
\\
32 &
\numprint{6070} & \numprint{5936} & \numprint{16.38}    &
\numprint{6064} & \numprint{5949} & \numprint{17.15}    &
\numprint{6189} & \numprint{6055} & \numprint{7.67}    &
\numprint{6082} & \numprint{5948} & \numprint{15.42}    &
\numprint{6110} & \numprint{5977} & \numprint{13.17}    
\\
\hline

\hline
\end{tabular}

\vspace*{0.3cm}

\begin{tabular}{|r|rr|rr|rr|rr|rr|rr|rr|rr|rr| } 
\hline
Effectiveness & \multicolumn{2}{c}{Strong} & \multicolumn{2}{c}{-Kway}  & \multicolumn{2}{c}{-Multitry} & \multicolumn{2}{c}{-MB}& \multicolumn{2}{c|}{-Flows}  \\
\hline
 $k$& Avg. & Best.  & Avg. & Best.  & Avg. & Best.  & Avg. & Best. & Avg. & Best.   \\
\hline 
2 &
\numprint{550} & \numprint{547} &
\numprint{550} & \numprint{548} &
\numprint{550} & \numprint{548} &
\numprint{550} & \numprint{548} &
\numprint{560} & \numprint{556} 
\\
4 &
\numprint{1251} & \numprint{1240} &
\numprint{1251} & \numprint{1243} &
\numprint{1254} & \numprint{1243} &
\numprint{1251} & \numprint{1241} &
\numprint{1266} & \numprint{1252} 
\\
8 &
\numprint{2263} & \numprint{2242} &
\numprint{2270} & \numprint{2249} &
\numprint{2276} & \numprint{2262} &
\numprint{2270} & \numprint{2246} &
\numprint{2281} & \numprint{2259} 
\\
16 &
\numprint{3771} & \numprint{3742} &
\numprint{3767} & \numprint{3741} &
\numprint{3810} & \numprint{3781} &
\numprint{3773} & \numprint{3747} &
\numprint{3797} & \numprint{3767} 
\\
32 &
\numprint{6000} & \numprint{5943} &
\numprint{6002} & \numprint{5950} &
\numprint{6090} & \numprint{6055} &
\numprint{6006} & \numprint{5955} &
\numprint{6028} & \numprint{5977} 
\\
64 &
\numprint{9523} & \numprint{9463} &
\numprint{9502} & \numprint{9437} &
\numprint{9681} & \numprint{9636} &
\numprint{9525} & \numprint{9470} &
\numprint{9548} & \numprint{9494} \\
\hline
Avg. & 
\numprint{2636} & \numprint{2616} &  
\numprint{2636} & \numprint{2618} & 
\numprint{2658} & \numprint{2642} & 
\numprint{2639} & \numprint{2619} & 
\numprint{2659} & \numprint{2637}

\\
\hline

\hline
\end{tabular}

\end{center}
\caption{Knockout tests: each configuration is the same as KaFFPa Strong minus the component shown at the top of the column. The first table shows detailed results for all $k$ in a normal test. The second table shows the results for an effectivity test.}
\end{table*}
\begin{table*}[h!] 
\small
\begin{center}

\begin{tabular}{|r|rrr|rrr|rrr|rrr|rrr|rrr| } 
\hline
$k$ & \multicolumn{3}{c}{Strong} & \multicolumn{3}{c}{-Kway}  & \multicolumn{3}{c}{-Multitry} & \multicolumn{3}{c}{-MB}& \multicolumn{3}{c|}{-Flows}  \\
\hline

 & Avg. & Best. & $t$ & Avg. & Best. & $t$ & Avg. & Best. & $t$ & Avg. & Best. & $t$& Avg. & Best. & $t$  \\
\hline 
2  & \numprint{561}  & \numprint{548}  & \numprint{2.85}  & \numprint{0,00}  & \numprint{0,00}  & \numprint{2,86}  & \numprint{0,00} & \numprint{0,00}  & \numprint{2,72} & \numprint{0,53} & \numprint{0,00}  & \numprint{2,70}  & \numprint{3,74} & \numprint{2,01} & \numprint{1,94} \\
4  & \numprint{1286} & \numprint{1242} & \numprint{5.14}  & \numprint{0,08}  & \numprint{-0,48} & \numprint{5,29}  & \numprint{0,54} & \numprint{-0,16} & \numprint{4,23} & \numprint{0,31} & \numprint{-0,24} & \numprint{4,68}  & \numprint{2,02} & \numprint{0,81} & \numprint{2,95} \\
8  & \numprint{2314} & \numprint{2244} & \numprint{7.52}  & \numprint{0,00}  & \numprint{-0,13} & \numprint{7,81}  & \numprint{0,99} & \numprint{1,20}  & \numprint{5,24} & \numprint{0,35} & \numprint{0,22}  & \numprint{6,88}  & \numprint{1,43} & \numprint{1,16} & \numprint{4,88} \\
16 & \numprint{3833} & \numprint{3746} & \numprint{11.19} & \numprint{-0,10} & \numprint{-0,29} & \numprint{11,69} & \numprint{1,59} & \numprint{1,41}  & \numprint{6,27} & \numprint{0,13} & \numprint{0,03}  & \numprint{10,41} & \numprint{0,97} & \numprint{0,88} & \numprint{8,22} \\
32 & \numprint{6070} & \numprint{5936} & \numprint{16.38} & \numprint{-0,10} & \numprint{0,22}  & \numprint{17,15} & \numprint{1,96} & \numprint{2,00}  & \numprint{7,67} & \numprint{0,20} & \numprint{0,20}  & \numprint{15,42} & \numprint{0,66} & \numprint{0,69} & \numprint{13,17} \\
\hline

\hline
\end{tabular}

\vspace*{0.3cm}

\begin{tabular}{|r|rr|rr|rr|rr|rr|rr|rr|rr|rr| } 
\hline
Effectiveness & \multicolumn{2}{c}{Strong} & \multicolumn{2}{c}{-Kway}  & \multicolumn{2}{c}{-Multitry} & \multicolumn{2}{c}{-MB}& \multicolumn{2}{c|}{-Flows}  \\
\hline
 $k$& Avg. & Best.  & Avg. & Best.  & Avg. & Best.  & Avg. & Best. & Avg. & Best.   \\
\hline 
2      & \numprint{550}  & \numprint{547}  & \numprint{0.00}  & \numprint{0.18}  & \numprint{0.00} & \numprint{0.18} & \numprint{0.00} & \numprint{0.18} & \numprint{1.82} & \numprint{1.65} \\
4      & \numprint{1251} & \numprint{1240} & \numprint{0.00}  & \numprint{0.24}  & \numprint{0.24} & \numprint{0.24} & \numprint{0.00} & \numprint{0.08} & \numprint{1.20} & \numprint{0.97} \\
8      & \numprint{2263} & \numprint{2242} & \numprint{0.31}  & \numprint{0.31}  & \numprint{0.57} & \numprint{0.89} & \numprint{0.31} & \numprint{0.18} & \numprint{0.80} & \numprint{0.76} \\
16     & \numprint{3771} & \numprint{3742} & \numprint{-0.11} & \numprint{-0.03} & \numprint{1.03} & \numprint{1.04} & \numprint{0.05} & \numprint{0.13} & \numprint{0.69} & \numprint{0.67} \\
32     & \numprint{6000} & \numprint{5943} & \numprint{0.03}  & \numprint{0.12}  & \numprint{1.50} & \numprint{1.88} & \numprint{0.10} & \numprint{0.20} & \numprint{0.47} & \numprint{0.57} \\
64     & \numprint{9523} & \numprint{9463} & \numprint{-0.22} & \numprint{-0.27} & \numprint{1.66} & \numprint{1.83} & \numprint{0.02} & \numprint{0.07} & \numprint{0.26} & \numprint{0.33} \\
\hline
Avg.   & \numprint{2636} & \numprint{2616} & \numprint{0.00}  & \numprint{0.08}  & \numprint{0.83} & \numprint{0.99} & \numprint{0.11} & \numprint{0.11} & \numprint{0.87} & \numprint{0.80} \\
\hline

\hline
\end{tabular}

\end{center}
\caption{Knockout tests: each configuration is the same as KaFFPa Strong minus the component shown at the top of the column. The first table shows detailed results for all $k$ in a normal test. The second table shows the results for an effectivity test. All values are increases in cut relative to the values obtained by KaFFPa Strong.}
\label{tab:knockoutexperiments}
\end{table*}

\begin{landscape}
\begin{table}[H]
\begin{center}
\tiny
\hspace*{-1.2cm}
\begin{tabular}{l|r|rrr|rrr|rrr|rrr|rrr|rrr|rrr|rrr|}
&& \multicolumn{3}{c|}{KaFFPa Strong} & \multicolumn{3}{c|}{KaFFPa Eco} & \multicolumn{3}{c|}{KaFFPa Fast} & \multicolumn{3}{c|}{KaSPar Strong} & \multicolumn{3}{c|}{KaPPa Strong} & \multicolumn{3}{c|}{DiBaP} & \multicolumn{3}{c|}{Scotch} & \multicolumn{3}{c|}{Metis}  \\
        graph & $k$ & Best & Avg. & $t$ & Best & Avg. & $t$ & Best & Avg. & $t$ &Best & Avg. & $t$ & Best & Avg. & $t$& Best & Avg. & $t$& Best & Avg. & $t$ & Best & Avg. & $t$\\
      \hline
      fe\_tooth 	 & 	 2 & 	 \numprint{3789} & 	 \numprint{3829} & 	  \numprint{5.43} & 	 \numprint{4159} & 	 \numprint{4594} & 	 \numprint{0.13} & 	 \numprint{4308} & 	 \numprint{4491} & 	  \numprint{0.12} & 	 \numprint{3844} & 	 \numprint{3987} & 	 \numprint{5.86} & 	 \numprint{3951} & 	 \numprint{4336} & 	  \numprint{3.75} & 	 \numprint{4390} & 	 \numprint{4785} & 	 \numprint{0.99} & 	 \numprint{3945} & 	 \numprint{4312} & 	  \numprint{0.36} & 	 \numprint{4319} & 	 \numprint{4695} & 	 \numprint{0.09} \\
fe\_tooth 	 & 	 4 & 	 \numprint{6812} & 	 \numprint{6946} & 	  \numprint{12.62} & 	 \numprint{7378} & 	 \numprint{7438} & 	 \numprint{0.38} & 	 \numprint{8047} & 	 \numprint{8773} & 	  \numprint{0.13} & 	 \numprint{6937} & 	 \numprint{6999} & 	 \numprint{8.54} & 	 \numprint{7012} & 	 \numprint{7189} & 	  \numprint{5.22} & 	 \numprint{7492} & 	 \numprint{8081} & 	 \numprint{1.11} & 	 \numprint{7464} & 	 \numprint{7770} & 	  \numprint{0.66} & 	 \numprint{7853} & 	 \numprint{8155} & 	 \numprint{0.10} \\
\hline
                 \end{tabular}\\

\end{center}
\caption{Detailed per instance basis results for the large testset.}
\label{tab:detailedperinstancebasis}
\end{table}
\end{landscape}

\begin{table*}[h!]
\begin{center}
        \small

\begin{tabular}{|r|rrr|rrr|rrr|rrr|}
\hline
         & \multicolumn{3}{c}{KaFFPa Strong}                                                                    & \multicolumn{3}{c}{KaFFPa Eco}                                                                   & \multicolumn{3}{c}{KaFFPa Fast}                                             & \multicolumn{3}{c|}{KaSPar Strong}                       \\
  \hline
  $k$    & Best.                             & Avg.                           & $t$                             & Best.                             & Avg.                             & $t$                       & Best.                      & Avg.                         & $t$             & Best.            & Avg.             & $t$              \\ 
  \hline
  2      & \numprint{3988}                   & \numprint{4001}                & \numprint{22.68}                & \numprint{4117}                   & \numprint{4178}                  & \numprint{0.79}           & \numprint{4573}            & \numprint{4459}              & \numprint{0.40} & \numprint{4013}  & \numprint{4047}  & \numprint{24.94}  \\
  4      & \numprint{10467}                  & \numprint{10559}               & \numprint{50.18}                & \numprint{10878}                  & \numprint{10969}                 & \numprint{1.42}           & \numprint{11897}           & \numprint{12732}             & \numprint{0.43} & \numprint{10548} & \numprint{10610} & \numprint{32.09}  \\
  8      & \numprint{19288}                  & \numprint{19553}               & \numprint{76.39}                & \numprint{20612}                  & \numprint{21061}                 & \numprint{2.06}           & \numprint{23026}           & \numprint{24295}             & \numprint{0.50} & \numprint{19332} & \numprint{19507} & \numprint{44.11}  \\
\hline
           
        \end{tabular}


	\vspace*{1cm}
 \begin{tabular}{|r|rrr|rrr|rrr|rrr|}
\hline
         & \multicolumn{3}{c}{KaFFPa Strong}                                                                    & \multicolumn{3}{c}{KaFFPa Eco}                                                                   & \multicolumn{3}{c}{KaFFPa Fast}                                             & \multicolumn{3}{c|}{KaSPar Strong}                       \\
  \hline
  $k$    & Best.                             & Avg.                           & $t$                             & Best.                             & Avg.                             & $t$                       & Best.                      & Avg.                         & $t$             & Best.            & Avg.             & $t$              \\ 
  \hline
  2      & \numprint{2812}                   & \numprint{2828}                & \numprint{31.44}                & \numprint{2925}                   & \numprint{2966}                  & \numprint{1.16}           & \numprint{3276}            & \numprint{3382}              & \numprint{0.55} & \numprint{2842}  & \numprint{2873}  & \numprint{36.89}  \\
  4      & \numprint{5636}                  &  \numprint{5709}               & \numprint{87.25}                & \numprint{5891}                  & \numprint{5996}                 & \numprint{2.83}           & \numprint{6829}               & \numprint{7408}              & \numprint{0.80} & \numprint{5642} & \numprint{5707} & \numprint{60.66}  \\
  8      & \numprint{10369}                  & \numprint{10511}               & \numprint{123.31}                & \numprint{11111}                  & \numprint{11398}                 & \numprint{3.82}           & \numprint{13149}           & \numprint{13856}            & \numprint{0.89} & \numprint{10464} & \numprint{10580} & \numprint{75.92}  \\
\hline
           
        \end{tabular}

\end{center}
\caption{Results for our large benchmark suite. The table on top contains average values for the comparison with DiBaP on our large testsuite without road networks and rgg.
         The table on the bottom contains average value for the comparisons with other general purpose partitioners on our large testsuite without the road network Europe for the case $k=2$.
         The average values are computed as described in Section~\ref{s:experiments}.}
\end{table*}
\begin{figure*}
\centering
\includegraphics[width=7cm]{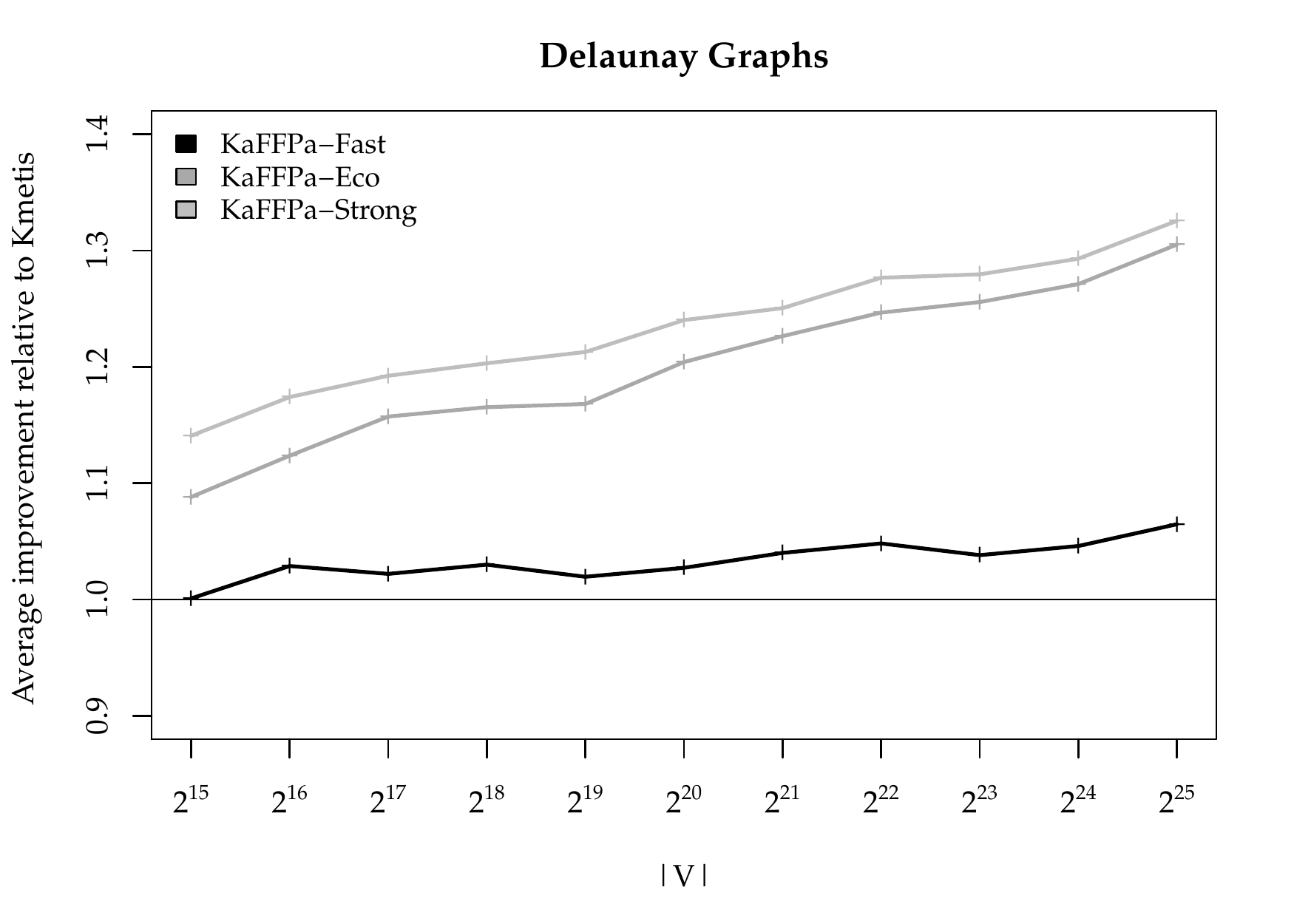}  \includegraphics[width=7cm]{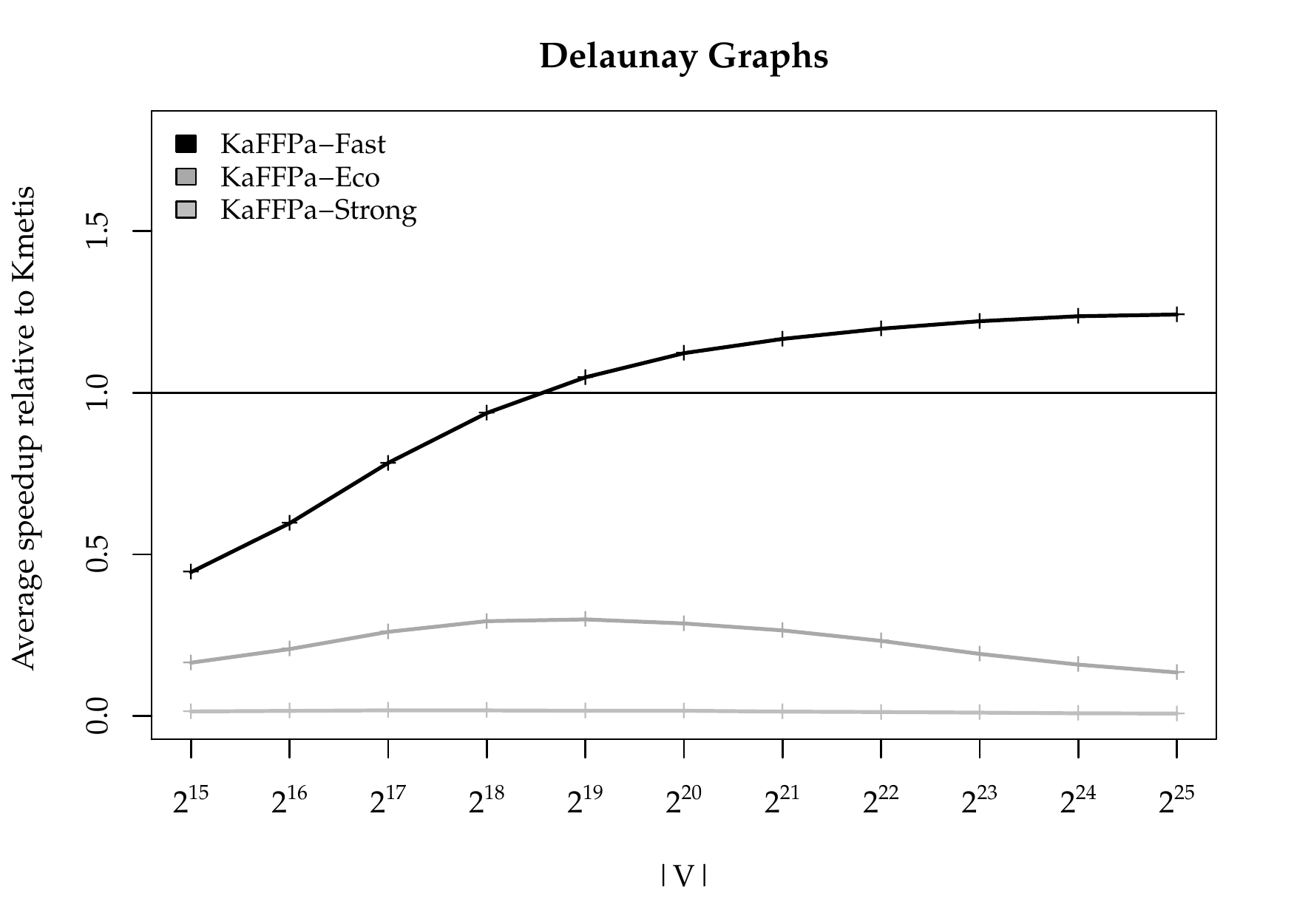}
\caption{Graph sequence test for Delaunay Graphs.}
\label{fig:}
\end{figure*}


\begin{landscape}
\begin{table}[H]
\begin{center}
\begin{tabular}{|l||r|r||r|r||r|r||r|r||r|r||r|r|}\hline

 Graph/$k$  & \multicolumn{2}{|c|}{2} & \multicolumn{2}{|c|}{4} & \multicolumn{2}{|c|}{8} & \multicolumn{2}{|c|}{16} & \multicolumn{2}{|c|}{32} & \multicolumn{2}{|c|}{64}\\

\hline 
3elt &       \textbf{\numprint{89}} & \numprint{89} & \textbf{\numprint{199}} & \numprint{199} & \textbf{\numprint{342}} & \numprint{342} & \numprint{571} & \textbf{\numprint{569}} & \numprint{987} & \textbf{\numprint{969}} & \numprint{1595} & \textbf{\numprint{1564}} \\
add20 &       \numprint{678} & \textbf{\numprint{594}} & \numprint{1197} & \textbf{\numprint{1177}} & \numprint{1740} & \textbf{\numprint{1704}} & \numprint{2156} & \textbf{\numprint{2121}} & \textbf{\numprint{2565}} & \numprint{2687} & \textbf{\numprint{3071}} & \numprint{3236} \\
data &       \textbf{\numprint{188}} & \numprint{188} & \textbf{\numprint{378}} & \numprint{383} & \textbf{\numprint{659}} & \numprint{660} & \numprint{1170} & \textbf{\numprint{1162}} & \numprint{2002} & \textbf{\numprint{1865}} & \numprint{2954} & \textbf{\numprint{2885}} \\
uk &       \textbf{\numprint{19}} & \numprint{19} & \textbf{\numprint{40}} & \numprint{41} & \textbf{\numprint{82}} & \numprint{84} & \textbf{\numprint{150}} & \numprint{152} & \numprint{260} & \textbf{\numprint{258}} & \textbf{\numprint{431}} & \numprint{438} \\
add32 &       \textbf{\numprint{10}} & \numprint{10} & \textbf{\numprint{30}} & \numprint{33} & \textbf{\numprint{66}} & \numprint{66} & \textbf{\numprint{117}} & \numprint{117} & \textbf{\numprint{212}} & \numprint{212} & \numprint{498} & \textbf{\numprint{493}} \\
bcsstk33 &       \textbf{\numprint{10097}} & \numprint{10097} & \numprint{21556} & \textbf{\numprint{21508}} & \numprint{34183} & \textbf{\numprint{34178}} & \numprint{55447} & \textbf{\numprint{54860}} & \numprint{79324} & \textbf{\numprint{78132}} & \numprint{110656} & \textbf{\numprint{108505}} \\
whitaker3 &       \textbf{\numprint{126}} & \numprint{126} & \textbf{\numprint{380}} & \numprint{380} & \textbf{\numprint{655}} & \numprint{656} & \numprint{1105} & \textbf{\numprint{1093}} & \textbf{\numprint{1700}} & \numprint{1717} & \numprint{2588} & \textbf{\numprint{2567}} \\
crack &       \textbf{\numprint{183}} & \numprint{183} & \textbf{\numprint{362}} & \numprint{362} & \textbf{\numprint{677}} & \numprint{678} & \numprint{1109} & \textbf{\numprint{1092}} & \numprint{1720} & \textbf{\numprint{1707}} & \numprint{2620} & \textbf{\numprint{2566}} \\
wing\_nodal &       \textbf{\numprint{1695}} & \numprint{1696} & \numprint{3576} & \textbf{\numprint{3572}} & \numprint{5445} & \textbf{\numprint{5443}} & \textbf{\numprint{8417}} & \numprint{8422} & \numprint{12129} & \textbf{\numprint{11980}} & \numprint{16332} & \textbf{\numprint{16134}} \\
fe\_4elt2 &       \textbf{\numprint{130}} & \numprint{130} & \textbf{\numprint{349}} & \numprint{349} & \textbf{\numprint{605}} & \numprint{605} & \textbf{\numprint{1006}} & \numprint{1014} & \textbf{\numprint{1647}} & \numprint{1657} & \numprint{2575} & \textbf{\numprint{2537}} \\
vibrobox &       \numprint{11538} & \textbf{\numprint{10310}} & \textbf{\numprint{19155}} & \numprint{19199} & \numprint{24702} & \textbf{\numprint{24553}} & \numprint{34384} & \textbf{\numprint{32167}} & \numprint{42711} & \textbf{\numprint{41399}} & \numprint{49924} & \textbf{\numprint{49521}} \\
bcsstk29 &       \textbf{\numprint{2818}} & \numprint{2818} & \numprint{8070} & \textbf{\numprint{8035}} & \numprint{14291} & \textbf{\numprint{13965}} & \numprint{23280} & \textbf{\numprint{21768}} & \numprint{36125} & \textbf{\numprint{34886}} & \numprint{58613} & \textbf{\numprint{57054}} \\
4elt &       \textbf{\numprint{138}} & \numprint{138} & \textbf{\numprint{320}} & \numprint{321} & \textbf{\numprint{534}} & \numprint{534} & \textbf{\numprint{938}} & \numprint{939} & \numprint{1576} & \textbf{\numprint{1559}} & \numprint{2623} & \textbf{\numprint{2596}} \\
fe\_sphere &       \textbf{\numprint{386}} & \numprint{386} & \textbf{\numprint{766}} & \numprint{768} & \textbf{\numprint{1152}} & \numprint{1152} & \textbf{\numprint{1710}} & \numprint{1730} & \textbf{\numprint{2520}} & \numprint{2565} & \numprint{3670} & \textbf{\numprint{3663}} \\
cti &       \textbf{\numprint{318}} & \numprint{318} & \textbf{\numprint{944}} & \numprint{944} & \textbf{\numprint{1752}} & \numprint{1802} & \textbf{\numprint{2865}} & \numprint{2906} & \textbf{\numprint{4180}} & \numprint{4223} & \numprint{6016} & \textbf{\numprint{5875}} \\
memplus &       \numprint{5596} & \textbf{\numprint{5489}} & \numprint{9805} & \textbf{\numprint{9559}} & \numprint{12126} & \textbf{\numprint{11785}} & \numprint{13564} & \textbf{\numprint{13241}} & \numprint{15232} & \textbf{\numprint{14395}} & \numprint{17595} & \textbf{\numprint{16857}} \\
cs4 &       \textbf{\numprint{366}} & \numprint{367} & \textbf{\numprint{938}} & \numprint{940} & \textbf{\numprint{1455}} & \numprint{1467} & \textbf{\numprint{2124}} & \numprint{2195} & \textbf{\numprint{2990}} & \numprint{3048} & \textbf{\numprint{4141}} & \numprint{4154} \\
bcsstk31 &       \textbf{\numprint{2699}} & \numprint{2701} & \textbf{\numprint{7296}} & \numprint{7444} & \textbf{\numprint{13274}} & \numprint{13371} & \numprint{24546} & \textbf{\numprint{24277}} & \numprint{38860} & \textbf{\numprint{38086}} & \numprint{60612} & \textbf{\numprint{60528}} \\
fe\_pwt &       \textbf{\numprint{340}} & \numprint{340} & \textbf{\numprint{704}} & \numprint{704} & \textbf{\numprint{1437}} & \numprint{1441} & \textbf{\numprint{2799}} & \numprint{2806} & \textbf{\numprint{5552}} & \numprint{5612} & \textbf{\numprint{8314}} & \numprint{8454} \\
bcsstk32 &       \textbf{\numprint{4667}} & \numprint{4667} & \textbf{\numprint{9208}} & \numprint{9247} & \numprint{21253} & \textbf{\numprint{20855}} & \textbf{\numprint{36968}} & \numprint{37372} & \numprint{62994} & \textbf{\numprint{61144}} & \numprint{97299} & \textbf{\numprint{95199}} \\
fe\_body &       \textbf{\numprint{262}} & \numprint{262} & \textbf{\numprint{598}} & \numprint{599} & \textbf{\numprint{1040}} & \numprint{1079} & \textbf{\numprint{1806}} & \numprint{1858} & \textbf{\numprint{2968}} & \numprint{3202} & \textbf{\numprint{5057}} & \numprint{5282} \\
t60k &       \textbf{\numprint{75}} & \numprint{75} & \textbf{\numprint{208}} & \numprint{211} & \textbf{\numprint{454}} & \numprint{465} & \textbf{\numprint{818}} & \numprint{849} & \textbf{\numprint{1361}} & \numprint{1391} & \textbf{\numprint{2143}} & \numprint{2211} \\
wing &       \textbf{\numprint{784}} & \numprint{787} & \textbf{\numprint{1616}} & \numprint{1666} & \textbf{\numprint{2509}} & \numprint{2589} & \textbf{\numprint{3889}} & \numprint{4131} & \textbf{\numprint{5747}} & \numprint{5902} & \textbf{\numprint{7842}} & \numprint{8132} \\
brack2 &       \textbf{\numprint{708}} & \numprint{708} & \textbf{\numprint{3013}} & \numprint{3027} & \textbf{\numprint{7110}} & \numprint{7144} & \textbf{\numprint{11745}} & \numprint{11969} & \textbf{\numprint{17751}} & \numprint{17798} & \numprint{26766} & \textbf{\numprint{26557}} \\
finan512 &       \textbf{\numprint{162}} & \numprint{162} & \textbf{\numprint{324}} & \numprint{324} & \textbf{\numprint{648}} & \numprint{648} & \textbf{\numprint{1296}} & \numprint{1296} & \textbf{\numprint{2592}} & \numprint{2592} & \numprint{10752} & \textbf{\numprint{10560}} \\
\hline
fe\_tooth &       \textbf{\numprint{3815}} & \numprint{3819} & \textbf{\numprint{6870}} & \numprint{6938} & \textbf{\numprint{11492}} & \numprint{11650} & \textbf{\numprint{17592}} & \numprint{18115} & \textbf{\numprint{25695}} & \numprint{25977} & \textbf{\numprint{35722}} & \numprint{35980} \\
fe\_rotor &       \textbf{\numprint{2031}} & \numprint{2045} & \numprint{7538} & \textbf{\numprint{7405}} & \numprint{13032} & \textbf{\numprint{12959}} & \numprint{20888} & \textbf{\numprint{20773}} & \textbf{\numprint{32678}} & \numprint{32783} & \numprint{47980} & \textbf{\numprint{47461}} \\
598a &       \textbf{\numprint{2388}} & \numprint{2388} & \textbf{\numprint{7956}} & \numprint{7992} & \textbf{\numprint{16050}} & \numprint{16179} & \textbf{\numprint{25892}} & \numprint{26196} & \textbf{\numprint{40003}} & \numprint{40513} & \textbf{\numprint{57795}} & \numprint{59098} \\
fe\_ocean &       \textbf{\numprint{387}} & \numprint{387} & \textbf{\numprint{1831}} & \numprint{1856} & \textbf{\numprint{4140}} & \numprint{4251} & \textbf{\numprint{8035}} & \numprint{8276} & \textbf{\numprint{13224}} & \numprint{13660} & \textbf{\numprint{20828}} & \numprint{21548} \\
144 &       \textbf{\numprint{6478}} & \numprint{6479} & \numprint{15635} & \textbf{\numprint{15196}} & \textbf{\numprint{25281}} & \numprint{25455} & \textbf{\numprint{38221}} & \numprint{38940} & \textbf{\numprint{56897}} & \numprint{58126} & \textbf{\numprint{80451}} & \numprint{81145} \\
wave &       \textbf{\numprint{8665}} & \numprint{8682} & \textbf{\numprint{16881}} & \numprint{16891} & \textbf{\numprint{29124}} & \numprint{29207} & \textbf{\numprint{43027}} & \numprint{43697} & \textbf{\numprint{62567}} & \numprint{64198} & \textbf{\numprint{86127}} & \numprint{88863} \\
m14b &       \textbf{\numprint{3826}} & \numprint{3826} & \textbf{\numprint{12981}} & \numprint{13034} & \textbf{\numprint{25854}} & \numprint{25921} & \textbf{\numprint{42358}} & \numprint{42513} & \textbf{\numprint{67454}} & \numprint{67770} & \textbf{\numprint{99661}} & \numprint{101551} \\
auto &       \textbf{\numprint{9958}} & \numprint{10004} & \textbf{\numprint{26669}} & \numprint{26941} & \numprint{45892} & \textbf{\numprint{45731}} &\textbf{\numprint{77163}} & \numprint{77618} & \textbf{\numprint{121645}} & \numprint{123296} & \textbf{\numprint{174527}} & \numprint{175975} \\
 \hline

\end{tabular}
 \end{center} \caption{Computing partitions from scratch $\epsilon = 1$\%. In each $k$-column the results computed by KaFFPa are on the left and the current Walshaw cuts are presented on the right side. }
\end{table}
\end{landscape}

\begin{landscape}
\begin{table}[H]
\begin{center}
\begin{tabular}{|l||r|r||r|r||r|r||r|r||r|r||r|r|}\hline

 Graph/$k$  & \multicolumn{2}{|c|}{2} & \multicolumn{2}{|c|}{4} & \multicolumn{2}{|c|}{8} & \multicolumn{2}{|c|}{16} & \multicolumn{2}{|c|}{32} & \multicolumn{2}{|c|}{64}\\

\hline 

3elt &       \textbf{\numprint{87}} & \numprint{87} & \textbf{\numprint{198}} & \numprint{198} & \textbf{\numprint{335}} & \numprint{336} & \textbf{\numprint{563}} & \numprint{565} & \numprint{962} & \textbf{\numprint{958}} & \numprint{1558} & \textbf{\numprint{1542}} \\
add20 &       \numprint{702} & \textbf{\numprint{576}} & \numprint{1186} & \textbf{\numprint{1158}} & \numprint{1724} & \textbf{\numprint{1690}} & \numprint{2104} & \textbf{\numprint{2095}} & \textbf{\numprint{2490}} & \numprint{2493} & \textbf{\numprint{3035}} & \numprint{3152} \\
data &       \textbf{\numprint{185}} & \numprint{185} & \textbf{\numprint{369}} & \numprint{378} & \textbf{\numprint{640}} & \numprint{650} & \textbf{\numprint{1127}} & \numprint{1133} & \numprint{1846} & \textbf{\numprint{1802}} & \numprint{2922} & \textbf{\numprint{2809}} \\
uk &       \textbf{\numprint{18}} & \numprint{18} & \textbf{\numprint{39}} & \numprint{40} & \textbf{\numprint{78}} & \numprint{81} & \textbf{\numprint{141}} & \numprint{148} & \textbf{\numprint{245}} & \numprint{251} & \numprint{418} & \textbf{\numprint{414}} \\
add32 &       \textbf{\numprint{10}} & \numprint{10} & \textbf{\numprint{30}} & \numprint{33} & \textbf{\numprint{66}} & \numprint{66} & \textbf{\numprint{117}} & \numprint{117} & \textbf{\numprint{212}} & \numprint{212} & \numprint{496} & \textbf{\numprint{493}} \\
bcsstk33 &       \textbf{\numprint{10064}} & \numprint{10064} & \textbf{\numprint{20865}} & \numprint{21035} & \textbf{\numprint{34078}} & \numprint{34078} & \numprint{54847} & \textbf{\numprint{54510}} & \numprint{78129} & \textbf{\numprint{77672}} & \numprint{108668} & \textbf{\numprint{107012}} \\
whitaker3 &       \textbf{\numprint{126}} & \numprint{126} & \textbf{\numprint{378}} & \numprint{378} & \textbf{\numprint{652}} & \numprint{655} & \textbf{\numprint{1090}} & \numprint{1092} & \textbf{\numprint{1680}} & \numprint{1686} & \numprint{2539} & \textbf{\numprint{2535}} \\
crack &       \textbf{\numprint{182}} & \numprint{182} & \textbf{\numprint{360}} & \numprint{360} & \textbf{\numprint{673}} & \numprint{676} & \numprint{1086} & \textbf{\numprint{1082}} & \numprint{1692} & \textbf{\numprint{1679}} & \numprint{2561} & \textbf{\numprint{2553}} \\
wing\_nodal &       \textbf{\numprint{1678}} & \numprint{1680} & \textbf{\numprint{3545}} & \numprint{3561} & \textbf{\numprint{5374}} & \numprint{5401} & \textbf{\numprint{8315}} & \numprint{8316} & \numprint{11963} & \textbf{\numprint{11938}} & \numprint{16097} & \textbf{\numprint{15971}} \\
fe\_4elt2 &       \textbf{\numprint{130}} & \numprint{130} & \textbf{\numprint{342}} & \numprint{343} & \textbf{\numprint{597}} & \numprint{598} & \textbf{\numprint{996}} & \numprint{1007} & \textbf{\numprint{1621}} & \numprint{1633} & \textbf{\numprint{2513}} & \numprint{2527} \\
vibrobox &       \numprint{11538} & \textbf{\numprint{10310}} & \numprint{18975} & \textbf{\numprint{18778}} & \numprint{24268} & \textbf{\numprint{24171}} & \numprint{33721} & \textbf{\numprint{31516}} & \numprint{42159} & \textbf{\numprint{39592}} & \numprint{49270} & \textbf{\numprint{49123}} \\
bcsstk29 &       \textbf{\numprint{2818}} & \numprint{2818} & \numprint{7993} & \textbf{\numprint{7983}} & \numprint{13867} & \textbf{\numprint{13817}} & \numprint{22494} & \textbf{\numprint{21410}} & \numprint{34892} & \textbf{\numprint{34407}} & \numprint{56682} & \textbf{\numprint{55366}} \\
4elt &       \textbf{\numprint{137}} & \numprint{137} & \textbf{\numprint{319}} & \numprint{319} & \textbf{\numprint{523}} & \numprint{523} & \numprint{918} & \textbf{\numprint{914}} & \numprint{1539} & \textbf{\numprint{1537}} & \textbf{\numprint{2570}} & \numprint{2581} \\
fe\_sphere &       \textbf{\numprint{384}} & \numprint{384} & \textbf{\numprint{764}} & \numprint{764} & \textbf{\numprint{1152}} & \numprint{1152} & \textbf{\numprint{1705}} & \numprint{1706} & \numprint{2483} & \textbf{\numprint{2477}} & \numprint{3568} & \textbf{\numprint{3547}} \\
cti &       \textbf{\numprint{318}} & \numprint{318} & \textbf{\numprint{916}} & \numprint{917} & \textbf{\numprint{1714}} & \numprint{1716} & \textbf{\numprint{2773}} & \numprint{2778} & \textbf{\numprint{4029}} & \numprint{4132} & \textbf{\numprint{5683}} & \numprint{5763} \\
memplus &       \numprint{5466} & \textbf{\numprint{5355}} & \numprint{9593} & \textbf{\numprint{9418}} & \numprint{12085} & \textbf{\numprint{11628}} & \numprint{13384} & \textbf{\numprint{13130}} & \numprint{15124} & \textbf{\numprint{14264}} & \numprint{17183} & \textbf{\numprint{16724}} \\
cs4 &       \textbf{\numprint{360}} & \numprint{361} & \textbf{\numprint{928}} & \numprint{936} & \textbf{\numprint{1439}} & \numprint{1467} & \textbf{\numprint{2090}} & \numprint{2126} & \textbf{\numprint{2935}} & \numprint{3014} & \textbf{\numprint{4080}} & \numprint{4107} \\
bcsstk31 &       \textbf{\numprint{2676}} & \numprint{2676} & \textbf{\numprint{7150}} & \numprint{7181} & \textbf{\numprint{13020}} & \numprint{13246} & \numprint{23536} & \textbf{\numprint{23504}} & \numprint{38048} & \textbf{\numprint{37459}} & \numprint{58738} & \textbf{\numprint{58667}} \\
fe\_pwt &       \textbf{\numprint{340}} & \numprint{340} & \textbf{\numprint{700}} & \numprint{704} & \textbf{\numprint{1411}} & \numprint{1416} & \textbf{\numprint{2776}} & \numprint{2784} & \textbf{\numprint{5496}} & \numprint{5606} & \textbf{\numprint{8228}} & \numprint{8346} \\
bcsstk32 &       \textbf{\numprint{4667}} & \numprint{4667} & \textbf{\numprint{8742}} & \numprint{8778} & \numprint{20223} & \textbf{\numprint{20035}} & \textbf{\numprint{35572}} & \numprint{35788} & \numprint{60766} & \textbf{\numprint{59824}} & \textbf{\numprint{92094}} & \numprint{92690} \\
fe\_body &       \textbf{\numprint{262}} & \numprint{262} & \textbf{\numprint{598}} & \numprint{598} & \textbf{\numprint{1016}} & \numprint{1033} & \textbf{\numprint{1734}} & \numprint{1767} & \textbf{\numprint{2810}} & \numprint{2906} & \textbf{\numprint{4799}} & \numprint{4982} \\
t60k &       \textbf{\numprint{71}} & \numprint{71} & \textbf{\numprint{203}} & \numprint{207} & \textbf{\numprint{449}} & \numprint{454} & \textbf{\numprint{805}} & \numprint{822} & \textbf{\numprint{1343}} & \numprint{1391} & \textbf{\numprint{2115}} & \numprint{2198} \\
wing &       \textbf{\numprint{773}} & \numprint{774} & \textbf{\numprint{1605}} & \numprint{1636} & \textbf{\numprint{2471}} & \numprint{2551} & \textbf{\numprint{3862}} & \numprint{4015} & \textbf{\numprint{5645}} & \numprint{5832} & \textbf{\numprint{7727}} & \numprint{8043} \\
brack2 &       \textbf{\numprint{684}} & \numprint{684} & \textbf{\numprint{2834}} & \numprint{2839} & \textbf{\numprint{6871}} & \numprint{6980} & \textbf{\numprint{11462}} & \numprint{11622} & \textbf{\numprint{17211}} & \numprint{17491} & \textbf{\numprint{26026}} & \numprint{26366} \\
finan512 &       \textbf{\numprint{162}} & \numprint{162} & \textbf{\numprint{324}} & \numprint{324} & \textbf{\numprint{648}} & \numprint{648} & \textbf{\numprint{1296}} & \numprint{1296} & \textbf{\numprint{2592}} & \numprint{2592} & \numprint{10629} & \textbf{\numprint{10560}} \\
\hline
fe\_tooth &       \textbf{\numprint{3788}} & \numprint{3792} & \textbf{\numprint{6796}} & \numprint{6862} & \textbf{\numprint{11313}} & \numprint{11422} & \textbf{\numprint{17318}} & \numprint{17655} & \textbf{\numprint{25208}} & \numprint{25624} & \textbf{\numprint{35044}} & \numprint{35830} \\
fe\_rotor &       \textbf{\numprint{1959}} & \numprint{1960} & \textbf{\numprint{7128}} & \numprint{7182} & \textbf{\numprint{12479}} & \numprint{12546} & \numprint{20397} & \textbf{\numprint{20356}} & \textbf{\numprint{31345}} & \numprint{31763} & \textbf{\numprint{46783}} & \numprint{47049} \\
598a &       \textbf{\numprint{2367}} & \numprint{2367} & \textbf{\numprint{7842}} & \numprint{7873} & \textbf{\numprint{15740}} & \numprint{15820} & \textbf{\numprint{25704}} & \numprint{25927} & \textbf{\numprint{38803}} & \numprint{39525} & \textbf{\numprint{57070}} & \numprint{58101} \\
fe\_ocean &       \textbf{\numprint{311}} & \numprint{311} & \textbf{\numprint{1696}} & \numprint{1698} & \textbf{\numprint{3921}} & \numprint{3974} & \textbf{\numprint{7648}} & \numprint{7838} & \textbf{\numprint{12550}} & \numprint{12746} & \textbf{\numprint{20049}} & \numprint{21033} \\
144 &       \textbf{\numprint{6438}} & \numprint{6438} & \numprint{15128} & \textbf{\numprint{15122}} & \textbf{\numprint{25119}} & \numprint{25301} & \textbf{\numprint{37782}} & \numprint{37899} & \textbf{\numprint{56399}} & \numprint{56463} & \textbf{\numprint{78626}} & \numprint{80621} \\
wave &       \textbf{\numprint{8594}} & \numprint{8616} & \textbf{\numprint{16668}} & \numprint{16822} & \textbf{\numprint{28513}} & \numprint{28664} & \textbf{\numprint{42308}} & \numprint{42620} & \textbf{\numprint{61756}} & \numprint{62281} & \textbf{\numprint{85254}} & \numprint{86663} \\
m14b &       \textbf{\numprint{3823}} & \numprint{3823} & \textbf{\numprint{12948}} & \numprint{12977} & \textbf{\numprint{25522}} & \numprint{25550} & \textbf{\numprint{42015}} & \numprint{42061} & \numprint{66401} & \textbf{\numprint{65879}} & \textbf{\numprint{96881}} & \numprint{100064} \\
auto &       \textbf{\numprint{9683}} & \numprint{9716} & \textbf{\numprint{25836}} & \numprint{25979} & \textbf{\numprint{44841}} & \numprint{45109} & \textbf{\numprint{75792}} & \numprint{76016} & \textbf{\numprint{120174}} & \numprint{120534} & \textbf{\numprint{171584}} & \numprint{172357} \\

 \hline

\end{tabular}
 \end{center} \caption{Computing partitions from scratch $\epsilon = 3$\%. In each $k$-column the results computed by KaFFPa are on the left and the current Walshaw cuts are presented on the right side. }
\end{table}
\end{landscape}

\begin{landscape}
\begin{table}[H]
\begin{center}
\begin{tabular}{|l||r|r||r|r||r|r||r|r||r|r||r|r|}\hline

 Graphi/$k$  & \multicolumn{2}{|c|}{2} & \multicolumn{2}{|c|}{4} & \multicolumn{2}{|c|}{8} & \multicolumn{2}{|c|}{16} & \multicolumn{2}{|c|}{32} & \multicolumn{2}{|c|}{64}\\

\hline 

3elt &       \textbf{\numprint{87}} & \numprint{87} & \textbf{\numprint{197}} & \numprint{197} & \textbf{\numprint{330}} & \numprint{330} & \textbf{\numprint{558}} & \numprint{560} & \numprint{952} & \textbf{\numprint{950}} & \textbf{\numprint{1528}} & \numprint{1539} \\
add20 &       \numprint{691} & \textbf{\numprint{550}} & \numprint{1171} & \textbf{\numprint{1157}} & \numprint{1703} & \textbf{\numprint{1675}} & \numprint{2112} & \textbf{\numprint{2081}} & \textbf{\numprint{2440}} & \numprint{2463} & \textbf{\numprint{2996}} & \numprint{3152} \\
data &       \numprint{182} & \textbf{\numprint{181}} & \textbf{\numprint{363}} & \numprint{368} & \numprint{629} & \textbf{\numprint{628}} & \numprint{1092} & \textbf{\numprint{1086}} & \numprint{1813} & \textbf{\numprint{1777}} & \numprint{2852} & \textbf{\numprint{2798}} \\
uk &       \textbf{\numprint{18}} & \numprint{18} & \textbf{\numprint{39}} & \numprint{39} & \textbf{\numprint{76}} & \numprint{78} & \textbf{\numprint{139}} & \numprint{139} & \textbf{\numprint{242}} & \numprint{246} & \textbf{\numprint{404}} & \numprint{410} \\
add32 &       \textbf{\numprint{10}} & \numprint{10} & \textbf{\numprint{30}} & \numprint{33} & \textbf{\numprint{63}} & \numprint{63} & \textbf{\numprint{117}} & \numprint{117} & \textbf{\numprint{212}} & \numprint{212} & \textbf{\numprint{486}} & \numprint{491} \\
bcsstk33 &       \textbf{\numprint{9914}} & \numprint{9914} & \numprint{20216} & \textbf{\numprint{20198}} & \textbf{\numprint{33922}} & \numprint{33938} & \numprint{54692} & \textbf{\numprint{54323}} & \numprint{77564} & \textbf{\numprint{77163}} & \numprint{107832} & \textbf{\numprint{106886}} \\
whitaker3 &       \textbf{\numprint{126}} & \numprint{126} & \textbf{\numprint{378}} & \numprint{378} & \textbf{\numprint{647}} & \numprint{650} & \numprint{1087} & \textbf{\numprint{1084}} & \textbf{\numprint{1673}} & \numprint{1686} & \textbf{\numprint{2512}} & \numprint{2535} \\
crack &       \textbf{\numprint{182}} & \numprint{182} & \textbf{\numprint{360}} & \numprint{360} & \textbf{\numprint{667}} & \numprint{667} & \textbf{\numprint{1077}} & \numprint{1080} & \numprint{1682} & \textbf{\numprint{1679}} & \textbf{\numprint{2526}} & \numprint{2548} \\
wing\_nodal &       \numprint{1669} & \textbf{\numprint{1668}} & \textbf{\numprint{3524}} & \numprint{3536} & \textbf{\numprint{5346}} & \numprint{5350} & \textbf{\numprint{8266}} & \numprint{8316} & \textbf{\numprint{11855}} & \numprint{11879} & \numprint{16111} & \textbf{\numprint{15873}} \\
fe\_4elt2 &       \textbf{\numprint{130}} & \numprint{130} & \textbf{\numprint{335}} & \numprint{335} & \textbf{\numprint{581}} & \numprint{583} & \textbf{\numprint{986}} & \numprint{991} & \textbf{\numprint{1600}} & \numprint{1633} & \textbf{\numprint{2493}} & \numprint{2516} \\
vibrobox &       \numprint{11486} & \textbf{\numprint{10310}} & \numprint{18856} & \textbf{\numprint{18778}} & \numprint{23948} & \textbf{\numprint{23930}} & \numprint{33113} & \textbf{\numprint{31235}} & \numprint{41812} & \textbf{\numprint{39592}} & \numprint{48841} & \textbf{\numprint{48200}} \\
bcsstk29 &       \textbf{\numprint{2818}} & \numprint{2818} & \numprint{7942} & \textbf{\numprint{7936}} & \textbf{\numprint{13575}} & \numprint{13614} & \numprint{21971} & \textbf{\numprint{20924}} & \numprint{34452} & \textbf{\numprint{33818}} & \numprint{55873} & \textbf{\numprint{54935}} \\
4elt &       \textbf{\numprint{137}} & \numprint{137} & \textbf{\numprint{315}} & \numprint{315} & \textbf{\numprint{516}} & \numprint{516} & \textbf{\numprint{901}} & \numprint{902} & \textbf{\numprint{1520}} & \numprint{1532} & \textbf{\numprint{2554}} & \numprint{2565} \\
fe\_sphere &       \textbf{\numprint{384}} & \numprint{384} & \textbf{\numprint{762}} & \numprint{764} & \textbf{\numprint{1152}} & \numprint{1152} & \textbf{\numprint{1688}} & \numprint{1692} & \textbf{\numprint{2433}} & \numprint{2477} & \textbf{\numprint{3535}} & \numprint{3547} \\
cti &       \textbf{\numprint{318}} & \numprint{318} & \textbf{\numprint{889}} & \numprint{890} & \textbf{\numprint{1684}} & \numprint{1708} & \numprint{2735} & \textbf{\numprint{2725}} & \textbf{\numprint{3957}} & \numprint{4037} & \textbf{\numprint{5609}} & \numprint{5684} \\
memplus &       \numprint{5362} & \textbf{\numprint{5267}} & \numprint{9690} & \textbf{\numprint{9299}} & \numprint{12078} & \textbf{\numprint{11555}} & \numprint{13349} & \textbf{\numprint{13078}} & \numprint{14992} & \textbf{\numprint{14170}} & \numprint{16758} & \textbf{\numprint{16454}} \\
cs4 &       \textbf{\numprint{353}} & \numprint{356} & \textbf{\numprint{922}} & \numprint{936} & \textbf{\numprint{1435}} & \numprint{1467} & \textbf{\numprint{2083}} & \numprint{2126} & \textbf{\numprint{2923}} & \numprint{2958} & \numprint{4055} & \textbf{\numprint{4052}} \\
bcsstk31 &       \textbf{\numprint{2670}} & \numprint{2676} & \textbf{\numprint{7088}} & \numprint{7099} & \textbf{\numprint{12865}} & \numprint{12941} & \textbf{\numprint{23202}} & \numprint{23254} & \textbf{\numprint{37282}} & \numprint{37459} & \numprint{57748} & \textbf{\numprint{57534}} \\
fe\_pwt &       \textbf{\numprint{340}} & \numprint{340} & \textbf{\numprint{700}} & \numprint{700} & \textbf{\numprint{1405}} & \numprint{1405} & \textbf{\numprint{2748}} & \numprint{2772} & \textbf{\numprint{5431}} & \numprint{5545} & \textbf{\numprint{8136}} & \numprint{8310} \\
bcsstk32 &       \textbf{\numprint{4622}} & \numprint{4622} & \textbf{\numprint{8441}} & \numprint{8454} & \textbf{\numprint{19601}} & \numprint{19678} & \textbf{\numprint{35014}} & \numprint{35208} & \textbf{\numprint{59456}} & \numprint{59824} & \numprint{91110} & \textbf{\numprint{91006}} \\
fe\_body &       \textbf{\numprint{262}} & \numprint{262} & \textbf{\numprint{589}} & \numprint{596} & \textbf{\numprint{1014}} & \numprint{1017} & \textbf{\numprint{1701}} & \numprint{1723} & \textbf{\numprint{2787}} & \numprint{2807} & \textbf{\numprint{4642}} & \numprint{4834} \\
t60k &       \textbf{\numprint{65}} & \numprint{65} & \textbf{\numprint{195}} & \numprint{196} & \textbf{\numprint{445}} & \numprint{454} & \textbf{\numprint{801}} & \numprint{818} & \textbf{\numprint{1337}} & \numprint{1376} & \textbf{\numprint{2106}} & \numprint{2168} \\
wing &       \textbf{\numprint{770}} & \numprint{770} & \textbf{\numprint{1597}} & \numprint{1636} & \textbf{\numprint{2456}} & \numprint{2528} & \textbf{\numprint{3842}} & \numprint{3998} & \textbf{\numprint{5586}} & \numprint{5806} & \textbf{\numprint{7651}} & \numprint{7991} \\
brack2 &       \textbf{\numprint{660}} & \numprint{660} & \textbf{\numprint{2731}} & \numprint{2739} & \textbf{\numprint{6634}} & \numprint{6671} & \textbf{\numprint{11240}} & \numprint{11358} & \textbf{\numprint{17137}} & \numprint{17256} & \textbf{\numprint{25827}} & \numprint{26281} \\
finan512 &       \textbf{\numprint{162}} & \numprint{162} & \textbf{\numprint{324}} & \numprint{324} & \textbf{\numprint{648}} & \numprint{648} & \textbf{\numprint{1296}} & \numprint{1296} & \textbf{\numprint{2592}} & \numprint{2592} & \numprint{10604} & \textbf{\numprint{10560}} \\
\hline
fe\_tooth &       \textbf{\numprint{3773}} & \numprint{3773} & \textbf{\numprint{6718}} & \numprint{6825} & \textbf{\numprint{11185}} & \numprint{11337} & \textbf{\numprint{17230}} & \numprint{17404} & \textbf{\numprint{24977}} & \numprint{25216} & \textbf{\numprint{34704}} & \numprint{35466} \\
fe\_rotor &       \textbf{\numprint{1940}} & \numprint{1950} & \textbf{\numprint{6999}} & \numprint{7045} & \textbf{\numprint{12353}} & \numprint{12380} & \textbf{\numprint{19935}} & \numprint{20132} & \textbf{\numprint{31016}} & \numprint{31450} & \textbf{\numprint{46006}} & \numprint{46608} \\
598a &       \textbf{\numprint{2336}} & \numprint{2336} & \textbf{\numprint{7738}} & \numprint{7763} & \textbf{\numprint{15502}} & \numprint{15544} & \textbf{\numprint{25560}} & \numprint{25585} & \textbf{\numprint{38884}} & \numprint{39144} & \textbf{\numprint{56586}} & \numprint{57412} \\
fe\_ocean &       \textbf{\numprint{311}} & \numprint{311} & \textbf{\numprint{1686}} & \numprint{1697} & \textbf{\numprint{3902}} & \numprint{3941} & \textbf{\numprint{7457}} & \numprint{7618} & \textbf{\numprint{12373}} & \numprint{12720} & \textbf{\numprint{19764}} & \numprint{20667} \\
144 &       \textbf{\numprint{6361}} & \numprint{6362} & \numprint{15321} & \textbf{\numprint{15122}} & \numprint{25078} & \textbf{\numprint{25025}} & \numprint{37505} & \textbf{\numprint{37433}} & \textbf{\numprint{56041}} & \numprint{56463} & \textbf{\numprint{78645}} & \numprint{79296} \\
wave &       \textbf{\numprint{8535}} & \numprint{8563} & \textbf{\numprint{16543}} & \numprint{16662} & \textbf{\numprint{28493}} & \numprint{28615} & \textbf{\numprint{42179}} & \numprint{42482} & \textbf{\numprint{61386}} & \numprint{61788} & \textbf{\numprint{84247}} & \numprint{85658} \\
m14b &       \textbf{\numprint{3802}} & \numprint{3802} & \textbf{\numprint{12945}} & \numprint{12976} & \textbf{\numprint{25151}} & \numprint{25292} & \textbf{\numprint{41538}} & \numprint{41750} & \textbf{\numprint{65087}} & \numprint{65231} & \textbf{\numprint{96580}} & \numprint{98005} \\
auto &       \textbf{\numprint{9450}} & \numprint{9450} & \textbf{\numprint{25310}} & \numprint{25399} & \textbf{\numprint{44360}} & \numprint{44520} & \numprint{75195} & \textbf{\numprint{75066}} & \textbf{\numprint{119125}} & \numprint{120001} & \textbf{\numprint{171355}} & \numprint{171459} \\
 \hline

\end{tabular}
 \end{center} \caption{Computing partitions from scratch $\epsilon = 5$\%. In each $k$-column the results computed by KaFFPa are on the left and the current Walshaw cuts are presented on the right side. }
\end{table}
\end{landscape}

\end{document}